\def\paperversion{arxiv}
\def\anonymoussubmission{1} \newif\ifsinglecolumn\singlecolumnfalse
\newif\ifwidemargins\widemarginsfalse
\newif\ifwarning\warningfalse
\newif\ifshowcomments\showcommentsfalse
\newif\ifblinded\blindedfalse
\newif\ifshowlinenums\showlinenumsfalse
\newif\ifreport\reportfalse
\newif\ifcopyrightspace\copyrightspacefalse
\newif\ifacknowledgments\acknowledgmentsfalse
\newif\ifshowpagenumbers\showpagenumberstrue
\newif\iffinalformat\finalformatfalse
\newif\ifweb\webfalse
\newif\ifexternalize\externalizetrue

\ifx\anonymoussubmission\undefined
\else
  \if\anonymoussubmission 1
    \blindedtrue
  \else
    \blindedfalse
  \fi
\fi
\IfFileExists{.blinded}{\blindedtrue}

\ifx\paperversion\xxxxundefined
\PackageError{paperversions}{*** No valid document version was specified.
Macro paperversions must be defined as one of (markup, draft,
local, submission, final, web, tr, trdraft, blindtr, trlocal, arxiv)}
\fi

\def\xxversion{\csname xx\paperversion\endcsname}
\newif\ifsawversion\sawversionfalse

\ifcase\xxversion\relax
\widemarginstrue
    \singlecolumntrue
    \warningtrue
    \showlinenumstrue
    \sawversiontrue
    \blindedfalse
\or \warningtrue
    \showlinenumsfalse
    \showcommentstrue
    \sawversiontrue
    \blindedfalse
\or \warningtrue
    \showcommentsfalse
    \blindedfalse
    \sawversiontrue
    \acknowledgmentstrue
    \showlinenumstrue
\or \warningtrue
    \showcommentsfalse
    \sawversiontrue
    \acknowledgmentstrue
    \showlinenumstrue
\or \sawversiontrue
    \showlinenumstrue
\or \blindedfalse
    \sawversiontrue
    \copyrightspacetrue
    \acknowledgmentstrue
    \showpagenumbersfalse
    \finalformattrue
\or \singlecolumntrue
    \blindedfalse
    \sawversiontrue
    \reporttrue
    \acknowledgmentstrue
    \webtrue
\or \blindedfalse
    \singlecolumntrue
    \showcommentstrue
    \sawversiontrue
    \reporttrue
    \showlinenumstrue
    \warningtrue
\or \blindedfalse
    \showcommentstrue
    \copyrightspacetrue
    \acknowledgmentstrue
    \sawversiontrue
    \finalformattrue
    \showlinenumstrue
    \warningtrue
\or \blindedfalse
    \sawversiontrue
    \copyrightspacetrue
    \acknowledgmentstrue
    \finalformattrue
    \webtrue
\or \singlecolumntrue
    \blindedtrue
    \acknowledgmentsfalse
    \sawversiontrue
    \reporttrue
    \webtrue
\or \reporttrue
    \warningtrue
    \showcommentsfalse
    \blindedfalse
    \sawversiontrue
    \acknowledgmentstrue
    \showlinenumstrue
\or \reporttrue
    \blindedfalse
    \sawversiontrue
    \copyrightspacetrue
    \acknowledgmentstrue
    \finalformattrue
    \externalizefalse
    \webtrue
\fi

\ifsawversion
\else
\fi

\let\xxversion=\undefined

\ifreport
\documentclass{llncs}
\else
    \iffinalformat
        \ifweb
\documentclass[runningheads]{llncs}
        \else
\documentclass[runningheads]{llncs}
        \fi
    \else
\ifblinded
            \documentclass[runningheads]{llncs}
        \else
            \documentclass[runningheads]{llncs}
        \fi
    \fi
\fi

\makeatletter
\@ifclassloaded{acmart}{
  \citestyle{acmauthoryear}
  \setcitestyle{nosort}
}{}
\makeatother

\newcommand*{\CiteKeyGeni}{geni-icfp2025}

\makeatletter
\@ifclassloaded{llncs}{

  \let\c@theorem\relax
  \let\c@lemma\relax
  \let\c@corollary\relax
  \let\c@definition\relax
  \let\c@example\relax
  \let\thetheorem\relax

  \ifblinded\else
    \ifweb
\newlength\pagewidth
    \newlength\pageheight
    \newlength\marginwidth
    \newlength\marginheight
    \setlength\pagewidth{17.2cm}
    \setlength\pageheight{24.0cm}
    \newlength\totalwidth
    \newlength\totalheight
    \setlength\totalwidth{12.2cm}
    \setlength\totalheight{19.3cm}
    \setlength\marginwidth{\dimexpr(\pagewidth-\totalwidth)/2\relax}
    \setlength\marginheight{\dimexpr(\pageheight-\totalheight)/2\relax}
    \usepackage[papersize={\pagewidth,\pageheight},total={\totalwidth,\totalheight},centering]{geometry}
    \else
  \fi
\usepackage{iftex}
  \usepackage{amssymb}
  \usepackage{amsmath,amsthm}
  \usepackage{mathtools}
  \ifluatex
    \usepackage{unicode-math}
    \setmathfont{Latin Modern Math} \else
\usepackage[T1]{fontenc}
    \usepackage{lmodern}
  \fi
  \usepackage{fancyhdr}
\let\sectionmark\@gobble
  \let\subsectionmark\@gobble
  \usepackage[numbers]{natbib}
  }{}
\makeatother

\ifluatex
\fi

\ifluatex
\else
\usepackage[scaled=0.90]{inconsolata}

\providecommand{\lBrack}{\llbracket}
\providecommand{\rBrack}{\rrbracket}
\providecommand{\mathbfit}[1]{\boldsymbol{#1}}
\fi

\usepackage{xspace}
\usepackage{amsfonts}
\usepackage{stmaryrd}
\ifluatex\else
\DeclareMathDelimiter{\lParen}{\mathopen}{stmry}{"4C}{stmry}{"4C}
\DeclareMathDelimiter{\rParen}{\mathclose}{stmry}{"4D}{stmry}{"4D}
\fi
\usepackage{graphicx,eso-pic}
\usepackage[export]{adjustbox}
\usepackage{nth}
\usepackage{siunitx}
\usepackage[inline]{enumitem}
\usepackage{ttquot}
\usepackage{placeins}
\usepackage{multirow}
\usepackage{multicol}
\usepackage{booktabs}
\usepackage{lipsum}
\usepackage{microtype}
\usepackage{textpos}
\usepackage{stackengine}
\usepackage{wrapfig}
\usepackage{cancel}
\usepackage{stackrel}
\usepackage{colortbl}
\BeforeBeginEnvironment{wrapfigure}{
    \setlength{\intextsep}{0.0ex}
    \setlength{\columnsep}{1.6ex}
}
\usepackage{tikz}
\usetikzlibrary{arrows.meta,backgrounds,positioning,matrix,calc,fit,shapes.symbols,decorations.pathreplacing,decorations.pathmorphing,shadows.blur,3d}
\usetikzlibrary{automata,arrows}
\usepackage{tikz-cd}
\usepackage[customcolors,shade]{hf-tikz}

\tikzset{
    tensorgrid/.style={
        matrix of math nodes,
        nodes={draw, minimum size=14pt, inner sep=1.5pt, anchor=center, fill=black!0, font=\scriptsize},
        column sep=-\pgflinewidth,
        row sep=-\pgflinewidth,
        nodes in empty cells,
    },
    tensorlabel/.style={font=\fontsize{7}{7}, color=black},
    highlight/.style={fill=tessa-gpu-jit-2!30},
    maskcolor/.style={fill=red!10},
    arrowstyle/.style={->, >=latex, thick}
}

\usepackage[subrefformat=simple,labelformat=simple]{subcaption}

\usepackage{etoolbox}

\newenvironment{centered}{\centering
}{\par
}

\newcommand{\sq}[1]{\vspace{#1}}
\newcommand{\sm}[1]{\vspace{-#1}}

\usepackage{hyperref}
\hypersetup{
  colorlinks,
  urlcolor=blue,
  citecolor=blue,
  filecolor=blue
  linkcolor=blue,
  allcolors=blue!70!black!70,
}
\numberwithin{equation}{section}
\usepackage{cleveref}
\crefformat{section}{\mbox{Section #2#1#3}}
\crefformat{appendix}{\mbox{Appendix #2#1#3}}
\crefformat{subsection}{\mbox{Section #2#1#3}}
\crefrangeformat{section}{\mbox{Sections #3#1#4--#5#2#6}}
\crefrangeformat{appendix}{\mbox{Appendices #3#1#4--#5#2#6}}
\crefrangeformat{subsection}{\mbox{Sections #3#1#4--#5#2#6}}
\crefrangeformat{subsubsection}{\mbox{Sections #3#1#4--#5#2#6}}
\crefmultiformat{section}{\mbox{Sections #2#1#3}}{ and \mbox{#2#1#3}}
{, \mbox{#2#1#3}}{, and \mbox{#2#1#3}}
\crefmultiformat{appendix}{\mbox{Appendices #2#1#3}}{ and \mbox{#2#1#3}}
{, \mbox{#2#1#3}}{, and \mbox{#2#1#3}}
\crefmultiformat{subsection}{\mbox{Sections #2#1#3}}{ and \mbox{#2#1#3}}
{, \mbox{#2#1#3}}{, and \mbox{#2#1#3}}
\crefmultiformat{subsubsection}{\mbox{Sections #2#1#3}}{ and \mbox{#2#1#3}}
{, \mbox{#2#1#3}}{, and \mbox{#2#1#3}}

  \crefname{table}{Table}{Table}

  \crefformat{figure}{Figure~#2#1#3}
  \crefrangeformat{figure}{\mbox{Figures #3#1#4--#5#2#6}}
  \crefmultiformat{figure}{\mbox{Figures #2#1#3}}{ and \mbox{#2#1#3}}
  {, \mbox{#2#1#3}}{, and \mbox{#2#1#3}}

\usepackage{tex-macros/formalisms}
\usepackage{tex-macros/notation}

\ifshowcomments
    \usepackage[inline]{aplcomments}
\else
    \usepackage[disabled]{aplcomments}
\fi

\definecolor{keyword-color}{HTML}{000000}
\definecolor{keyword-color-1}{HTML}{9465B6}
\definecolor{keyword-color-2}{HTML}{9465B6}
\definecolor{codecomment-color}{HTML}{000000}
\definecolor{string-color}{HTML}{AA99CC}
\usepackage{listings,lstautogobble}
\providecommand{\tessalstrm}{\rmfamily}
  \ifluatex
  \else
  \fi
  \let\origthelstnumber\thelstnumber
  \makeatletter
  \newcommand*\Suppressnumber{\lst@AddToHook{OnNewLine}{\let\thelstnumber\relax \advance\c@lstnumber-\@ne\relax }}
  \newcommand*\Reactivatenumber{\lst@AddToHook{OnNewLine}{\let\thelstnumber\origthelstnumber \advance\c@lstnumber\@ne\relax}}
\makeatother

\definecolor{codeemph-color}{HTML}{ECF5F8}
\newcommand\codeemph[2][]{\ifthenelse{\isempty{#1}}{\colorlet{color}{black}}{\colorlet{color}{#1}}\setlength\fboxsep{1pt}\colorbox{codeemph-color}{\textcolor{color}{\mbox{#2}}}}

\def\thetitle{Tensor Probabilistic Model Checking of Finite-Horizon Markov Chains\xspace}

\ifreport
\title{\thetitle \mbox{(Extended Version)}}
\else
\title{\thetitle}
\titlerunning{Tensor Probabilistic Model Checking of Finite-Horizon Markov Chains}
\fi

\ifreport
\fancypagestyle{tr-style}{
  \fancyhf{}
  \fancyhead[LO]{\scriptsize\thetitle (Extended Version)}
  \fancyhead[RO]{\scriptsize\thepage}
  \fancyhead[RE]{\scriptsize Jianlin Li, Nick Guo, Peter Ye, and Yizhou Zhang}
  \fancyhead[LE]{\scriptsize\thepage}
}
\fancypagestyle{tr-style-firstpage}{
  \fancyhf{}
\fancyfoot[L]{\scriptsize \raisebox{2ex}{\begin{minipage}{\textwidth}
    \textcopyright\ 2026 The Authors. \\
    Technical report extending the authors' CAV 2026 paper \cite{tessa-cav2026} with an appendix.
    \end{minipage}}
  }
}
\else
\ifweb
\fancypagestyle{web-style-firstpage}{
  \fancyhf{}
  \fancyfoot[L]{\scriptsize \raisebox{2ex}{\begin{minipage}{\textwidth}
    \textcopyright\ 2026 The Authors. \\
    Authors' version of their work accepted to appear in CAV 2026.
    \end{minipage}}
  }
}
\fi
\fi

\ifblinded
  \author{}
  \institute{}
\else
  \makeatletter
  \@ifclassloaded{llncs}{
    \usepackage{orcidlink}
    \usepackage{marvosym} \newcommand{\Orcid}[1]{\hspace{1pt}\orcidlink{#1}}
    \author{Jianlin Li\Orcid{0000-0001-7371-3034} \and Nick Guo\Orcid{0009-0004-0548-0143} \and Peter Ye\Orcid{0009-0000-1256-2930} \and Yizhou Zhang\textsuperscript{(\scalebox{1.2}{\Letter})}\!\Orcid{0000-0002-8206-4694}}
    \institute{David R.~Cheriton School of Computer Science\\ University of Waterloo, Canada\\
    \email{\{jianlin.li, nick.guo, p2ye, yizhou\}@uwaterloo.ca}
}
    \authorrunning{Jianlin Li, Nick Guo, Peter Ye, and Yizhou Zhang}
  }{
    \@ifclassloaded{acmart}{
      \author{Jianlin Li}
      \affiliation{\department{David R.\ Cheriton School of Computer Science}
        \institution{University of Waterloo}
        \state{Ontario}
        \country{Canada}
      }
      \orcid{0000-0001-7371-3034}
      \author{Nick Guo}
      \affiliation{\department{David R.\ Cheriton School of Computer Science}
        \institution{University of Waterloo}
        \state{Ontario}
        \country{Canada}
      }
      \orcid{0009-0004-0548-0143}
      \author{Peter Ye}
      \affiliation{\department{David R.\ Cheriton School of Computer Science}
        \institution{University of Waterloo}
        \state{Ontario}
        \country{Canada}
      }
      \orcid{0009-0000-1256-2930}
      \author{Yizhou Zhang}
      \affiliation{\department{David R.\ Cheriton School of Computer Science}
        \institution{University of Waterloo}
        \state{Ontario}
        \country{Canada}
      }
      \orcid{0000-0002-8206-4694}
    }{}
  }
  \makeatother
\fi

\makeatletter
\@ifclassloaded{acmart}{
  \authorsaddresses{}
}{}
\makeatother

\ifcopyrightspace
\else
  \makeatletter
  \@ifclassloaded{acmart}{
    \setcopyright{none}
    \renewcommand\footnotetextcopyrightpermission[1]{}
    \settopmatter{printacmref=false}
  }{}
  \makeatother
\fi

\ifwarning

\makeatletter
\AtBeginDocument{
  \@ifclassloaded{acmart}{
    \AddToShipoutPicture*{\put(305,650){\it\color{red!80!black}{\fbox{\large Draft---please do not distribute}}}}
  }{}
  \@ifclassloaded{llncs}{
\AddToShipoutPicture*{\put(3,582){\it\color{red!80!black}{\fbox{\large Draft---please do not distribute}}}}
  }{}
}
\makeatother
\fi

\newcommand{\Lang}{{\textrm{PML}}\xspace}

\newcommand{\storm}{\mbox{\textrm{Storm}}\xspace}
\newcommand{\tessa}{\mbox{\textrm{Tessa}}\xspace} \newcommand{\geni}{\mbox{\textrm{Geni}}\xspace}
\newcommand{\dice}{\mbox{\textrm{Dice}}\xspace}

\iffinalformat
\ifreport
\else
\ifweb
\makeatletter
\@ifclassloaded{llncs}{\patchcmd{\@maketitle}{\@author\vskip.35cm}{\@author\tikz[remember picture, overlay]{\node[inner sep=0pt] (authorsline) {};}\vskip.35cm}{}{\PackageWarning{tessa}{could not patch \string\@maketitle\space for AE badge placement}}}{}\AtBeginDocument{\@ifclassloaded{llncs}{\AddToShipoutPicture*{\begin{tikzpicture}[remember picture, overlay]
\node[anchor=north east, inner sep=0pt]
          at ([xshift=\dimexpr-\marginwidth-3.143bp\relax, yshift=0.00bp]current page.north east |- authorsline)
          {\href{https://doi.org/10.5281/zenodo.19802567}{\includegraphics{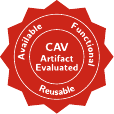}}};
      \end{tikzpicture}}}{}}
\makeatother
\fi
\fi
\fi

\newif\iftight
\tighttrue
 
\makeatletter
\@ifclassloaded{acmart}{
  \renewcommand\@parfont{\bfseries\sffamily}
  \renewcommand\noindentparagraph{\@startsection{paragraph}{4}{\z@}{-1.2ex}{-3.5\p@}{\ACM@NRadjust{\@parfont}}}
}{
  \newcommand\noindentparagraph[1]{\paragraph{#1}}
  \iftight
  \renewcommand\section{\@startsection{section}{1}{\z@}{-18\p@ \@plus -4\p@ \@minus -4\p@}{10\p@ \@plus 4\p@ \@minus 4\p@}{\normalfont\large\bfseries\boldmath
                        \rightskip=\z@ \@plus 8em\pretolerance=10000 }}
  \renewcommand\subsection{\@startsection{subsection}{2}{\z@}{-18\p@ \@plus -4\p@ \@minus -4\p@}{6\p@ \@plus 4\p@ \@minus 4\p@}{\normalfont\normalsize\bfseries\boldmath
                        \rightskip=\z@ \@plus 8em\pretolerance=10000 }}
  \fi
}
\makeatother

\newsavebox\crunchedbox
\def\vcrunch#1{\savebox\crunchedbox{#1}\smash{\usebox\crunchedbox}{\vrule width 0pt height 0.9em}}

\newcommand{\TODO}[1][]{\if\relax\detokenize{#1}\relax
    \colorbox{pure-apple}{\color{white}\texttt{TODO}}
  \else
    \colorbox{pure-apple}{\color{white}\texttt{{#1}}}
  \fi
}
\newcommenter{JL}{0.7,1.0,0.7}    

\usepackage{pgfplots}
\usepackage{pgfplotstable}
\pgfplotsset{compat=1.18}
\usepgfplotslibrary{colorbrewer}
\usepgfplotslibrary{fillbetween}
\usepgfplotslibrary{groupplots}
\pgfplotsset{
  every axis/.append style={
    cycle list/Dark2,
    axis lines=box,
    axis line style={-,black},
    label style={font=\fontsize{8}{8}\selectfont}, tick label style={font=\tiny\color{black}},
    xtick pos=left,
    ytick pos=left,
    enlarge x limits={.05},
    legend cell align=center,
title style={font=\fontsize{8}{8}\selectfont},
    legend style={font=\fontsize{7}{7}\selectfont,nodes={scale=1.0,},draw=none,}, legend cell align={left},
    legend image post style={scale=1.0},
    legend image code/.code={
      \draw[mark repeat=2,mark phase=2]
      plot coordinates {
      (0cm,0cm)
      (0.2cm,0cm)        (0.4cm,0cm)         };}
  },
  every axis plot/.append style={
    every mark/.append style={
      mark size=1.25pt,
    },
    thick,
  },
}
\pgfplotsset{select coords between index/.style 2 args={
  x filter/.code={
    \ifnum\coordindex<#1\fi
    \ifnum\coordindex>#2\fi
  }
}}
\pgfplotsset{
  discard if/.style 2 args={
    x filter/.code={
      \edef\tempa{\thisrow{#1}}
      \edef\tempb{#2}
      \ifx\tempa\tempb
      
      \fi
    }
  },
  discard if not/.style 2 args={
    x filter/.code={
      \edef\tempa{\thisrow{#1}}
      \edef\tempb{#2}
      \ifx\tempa\tempb
      \else
      
      \fi
    }
  },
  discard if not A/.style 2 args={
    x filter/.append code={
      \edef\tempa{\thisrow{#1}}
      \edef\tempaa{#2}
      \ifx\tempa\tempaa
      \else
        
      \fi
    }
  },
}

\definecolor{storm-add}{HTML}{BF2494}
\definecolor{storm-spm}{HTML}{FD6467} \definecolor{dice}{HTML}{6F9231}
\definecolor{geni}{HTML}{EE9722} \definecolor{tessa-gpu-jit-1}{HTML}{046C9A}
\colorlet{tessa-gpu-jit-2}{tessa-gpu-jit-1!50}

\begin{document}

\makeatletter
\@ifclassloaded{llncs}{
  \maketitle
}{}
\makeatother

\begin{abstract}

We reexamine the problem of verifying Markov chains with respect to
step-bounded reachability probabilities.
Prevailing approaches rely on encoding the
state-transition matrix using either explicit or symbolic
representations.
While these approaches are effective for sparse transition
dynamics, they scale less favorably in the dense regime.
\medskip

Our insight is to cast probabilistic model checking of Markov chains
as computations over dense tensors.
This methodology enables the use of
off-the-shelf compiler toolchains for optimized
execution of these tensor computations on hardware accelerators.
We prove the soundness of the methodology of mapping probabilistic
model checking to tensor computations.
We implement our approach in a tool called \tessa.
Empirical evaluation shows that \tessa unlocks
massive speedups over state-of-the-art methods
on selected benchmarks from the literature.
 \end{abstract}

\makeatletter
\@ifclassloaded{acmart}{
  \maketitle
}{}
\makeatother

\ifreport
\thispagestyle{tr-style-firstpage}
\else
\ifweb
\thispagestyle{web-style-firstpage}
\fi
\fi

\renewcommand{\Vvdash}{\cubeop}
\renewcommand{\Vdash}{\mathbin{\raisebox{0.28ex}{\scalebox{0.8}[0.5]{$\blacktriangleright$}}}}

\section{Introduction}
\label{tessa:sec:intro}

Probabilistic model checking is a formal verification technique for
verifying quantitative properties of systems exhibiting stochastic
behavior. Given a mathematical model (such as a Markov chain)
representing the behavior of a system over time, an algorithmic
procedure determines whether the system satisfies properties of
interest (such as step-bounded reachability).

\newcommand{\FacultyState}[1]{\textsf{#1}}
As an example, consider the problem of scheduling an urgent faculty
meeting (\cref{fig:motivating-example}) involving $N$
professors by collecting their availability via a Doodle poll.
At any time step, each professor $i$ is in one of three states:
\FacultyState{Away}, \FacultyState{Doodling}, or \FacultyState{Done}.
A professor \FacultyState{Away} from their inbox may,
with probability $p_i$, notice the Doodle poll email and
transition to \FacultyState{Doodling} to answer the Doodle poll.
However, the \FacultyState{Doodling} state is fragile: with probability
$q_i$, they successfully complete the poll and reach
\FacultyState{Done}; but with probability $1-q_i$, they are
interrupted by a new email or a knock on their door,
forcing them back to be \FacultyState{Away} without hitting submit.
The dynamics capture the intermittent nature of professors' attention.
Each Markov chain in \cref{fig:motivating-example-markov-chains}
depicts the behavior of a professor.

The quantitative property to verify is the probability that all $N$
professors have reached \FacultyState{Done} within a horizon of $H$
time steps.

The program in \cref{fig:motivating-example-prism-code} models this
system of stochastic processes,
similarly to how one would structure it
in the PRISM~\cite{prism-manual-v4.9}
and JANI~\cite{jani-tacas2017} modeling languages.
Each module $P_i$ models the behavior of professor $i$:
it declares a state variable $s_i$ with domain $\{0, 1, 2\}$
(indicating \FacultyState{Away}, \FacultyState{Doodling}, and
\FacultyState{Done}, respectively),
and three \emph{commands} that update $s_i$ according to the
transition probabilities $p_i$ and $q_i$.
Each command is \emph{guarded} by a Boolean expression
(e.g., $s_i=0$ for the first command);
a command is enabled only if its guard is satisfied
in the current state.
The \textbf{goal} clause of the global model specifies
the target property that all professors are in state
\FacultyState{Done}.

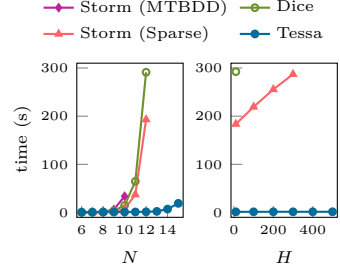
\begin{figure}[t]
\centering
\begin{subfigure}{\textwidth}
\fontsize{7.5}{8.5}\selectfont
\usetikzlibrary{automata, positioning, arrows}
\newcommand{\profmc}[1]{\scalebox{.7}{\begin{tikzpicture}[->, >=stealth, shorten >=1pt, auto, node distance=55pt, semithick]
\tikzstyle{every state}=[fill=white, draw=black, text=black, minimum size=40pt, align=center, inner sep=2pt]
  \tikzstyle{absorbing}=[double, double distance=1pt] 

\node[state, initial, initial text=] (res) {\FacultyState{Away}\\$s_{#1}=0$};
  \node[state] (wrestle) [right of=res, xshift=0pt] {\FacultyState{Doodling}\\$s_{#1}=1$};
  \node[state, absorbing] (sub) [right of=wrestle, xshift=0pt] {\FacultyState{Done}\\$s_{#1}=2$};

\path (res) edge [loop above, min distance=15pt] node [left, xshift=-4pt, yshift=-3pt] {$1-p_{#1}$} (res)
              edge [bend left=30] node {$p_{#1}$} (wrestle);

\path (wrestle) edge [bend left=30] node [below] {$1-q_{#1}$} (res)
                  edge node {$q_{#1}$} (sub);

\path (sub) edge [loop above, min distance=15pt] node [left, xshift=-5pt, yshift=-3pt] {$1$} (sub);

\end{tikzpicture}}}
\newcommand{\mcpar}{\raisebox{12pt}{\scalebox{3.0}{$\parallel$}}
}
\newcommand{\mcomit}{\raisebox{12pt}{\scalebox{3.0}{$\cdots$}}
}
\resizebox{\linewidth}{!}{\!\!\profmc{1} \mcpar \mcomit \mcpar \!\!\profmc{N}\ }
\caption{$N$ parallel Markov chains capturing the stochastic behavior of $N$ professors.}
\label{fig:motivating-example-markov-chains}
\end{subfigure}

\sq{2pt}

\begin{subfigure}[b]{0.615\textwidth}
\begin{lstlisting}
module [@$P_1$@]
   var [@$s_1$@] : [0 [@\Key{..}@] 2] init 0;
   [a] ([@$s_1=0$@]) -> [@$p_1$@] : ([@$s_1 \GETS 1$@]) + ([@$1-p_1$@]) : ([@$s_1 \GETS 0$@]);
   [a] ([@$s_1=1$@]) -> [@$q_1$@] : ([@$s_1 \GETS 2$@]) + ([@$1-q_1$@]) : ([@$s_1 \GETS 0$@]);
   [a] ([@$s_1=2$@]) -> 1.0 : ([@$s_1 \GETS 2$@]);
[@\sm{3pt}@]
module [@$\cdots$@] // modules [@$P_2$@] [@$\cdots$@] [@$P_N$@]
[@\sm{3pt}@]
goal [@$s_1 = 2 \wedge \cdots \wedge s_N=2$@]
\end{lstlisting}
\caption{A model of $N$ professors (in PRISM-like syntax).}
\label{fig:motivating-example-prism-code}
\end{subfigure}
\hfill
\begin{subfigure}[b]{0.37\textwidth}
    \centering
\begin{tikzpicture}
\catcode`\_=12
\def\FirstRunCol{work_seconds}
\def\SecondRunCol{measured_avg_seconds}
\def\BaselineCol{elapsed_seconds}
\def\StormAddCsv{storm.add.csv}
\def\StormSpmCsv{storm.spm.csv}
\def\RubiconCsv{rubicon.csv}
\def\TessaCsv{tessa.csv}
\def\TestN{tessa/camera-ready/meeting/testn/}
\def\TestH{tessa/camera-ready/meeting/testh/}
\begin{groupplot}[
  group style={
    group size=2 by 1,
    horizontal sep=18pt,
},
  height=103pt,
  width=85pt,
  legend style={
      transpose legend,
      legend columns=2,
      /tikz/every even column/.append style={column sep=-8pt},
      fill=none,
      font=\fontsize{6}{6}\selectfont,
},
xlabel style={font=\fontsize{6.5}{7}\selectfont},
  ymax=310,
  ymin=-10,
]

\nextgroupplot[
  xlabel={$N$},
legend to name=tessa:meeting-time-vs-N-A,
]

\addplot[storm-add, mark=diamond*]
table [
  col sep=comma,
  x=N,
  y=\BaselineCol,
  restrict expr to domain={x}{6:10},
  unbounded coords=discard,
]
{\TestN\StormAddCsv};
\addlegendentry{\storm (MTBDD)\ \ \ \ \ }

\addplot[storm-spm, mark=triangle*]
table [
  col sep=comma,
  x=N,
  y=\BaselineCol,
  restrict expr to domain={x}{6:12},
]
{\TestN\StormSpmCsv};
\addlegendentry{\storm (Sparse)}

\addplot[dice, mark=o]
table [
  col sep=comma,
  x=N,
  y=\BaselineCol,
  discard if not={H}{10},
  restrict expr to domain={x}{6:12},
]
{\TestN\RubiconCsv};
\addlegendentry{\dice}

\addplot[tessa-gpu-jit-1, mark=*]
table [
  col sep=comma,
  x=N,
  y=\FirstRunCol,
  discard if not={H}{10},
  restrict expr to domain={x}{6:15},
]
{\TestN\TessaCsv};
\addlegendentry{\tessa}

\nextgroupplot[
  xlabel={$H$},
  legend to name=tessa:meeting-time-vs-N-B,
]

\addlegendimage{storm-add, mark=diamond*} 

\addlegendentry{\storm (MTBDD)\ \ \ \ \ }

\addplot[dice, mark=o]
table [
  col sep=comma,
  x=H,
  y=\BaselineCol,
  discard if not={status}{ok},
  discard if not A={N}{12},
  restrict expr to domain={x == 10}{1:1},
]
{\TestH\RubiconCsv};
\addlegendentry{\dice}

\addplot[storm-spm, mark=triangle*]
table [
  col sep=comma,
  x=H,
  y=\BaselineCol,
  restrict expr to domain={y}{0:300},
restrict expr to domain={mod(x, 100)==0 && x<=500 || x==10}{1:1},
  unbounded coords=discard,
]
{\TestH\StormSpmCsv};
\addlegendentry{\storm (Sparse)}

\addplot[tessa-gpu-jit-1, mark=*]
table [
  x=H,
  y=\FirstRunCol,
  col sep=comma,
  discard if not={N}{12},
restrict expr to domain={mod(x, 100)==0 && x<=500 || x==10}{1:1},
  unbounded coords=discard,
]
{\TestH\TessaCsv};

\end{groupplot}

\coordinate (top) at (rel axis cs:0,1);\coordinate (bot) at (rel axis cs:1,0);

\node [above,inner sep=1pt,xshift=0pt] at (current bounding box.north) {\pgfplotslegendfromname{tessa:meeting-time-vs-N-A}};

\path (top)--(bot) coordinate[midway] (group center);
\node[above,rotate=90,xshift=0pt,yshift=-3pt] at (group center -| current bounding box.west) {
    \begin{minipage}{25pt}
    \centering
    \fontsize{6.5}{7}\selectfont
    time (s)
    \end{minipage}
};

\end{tikzpicture}\ignorespaces
     \caption{Scaling plots.}
    \label{fig:motivating-example-scaling-plot}
\end{subfigure}
\caption{Running example: collecting availability from $N$ professors. The target property is the probability that all $N$ professors reach \FacultyState{Done} within horizon~$H$.}
\label{fig:motivating-example}
\end{figure}

Despite the semantic simplicity of this model, the verification
problem is inherently computationally expensive due to state
explosion.
Since each process has 3 states, the global state-space size is $3^N$.
It is unlikely that any model checking tool on a classical computer
can avoid this exponential blowup.
Nevertheless, we would like to push the boundary of tractable
instances as far as possible,
especially given the raw speed of modern hardware accelerators
at our disposal.

State-of-the-art methods largely fall into two categories:
\begin{itemize}[align=left,labelsep=*,leftmargin=*,parsep=0pt,itemsep=2pt,topsep=2pt,]
\item
Probabilistic model checkers, such as PRISM \cite{prism4} and
Storm~\cite{storm2022}, are the prevailing tools for such verification
tasks. They represent the state-transition dynamics
using either explicit
(i.e., sparse matrices of size $3^N\,{\times}\,3^N$ in the $N$-professor model)
or symbolic
(i.e., multi-terminal binary decision diagrams, or MTBDDs) representations.
Both representations are highly optimized for CPU execution
in the situation of sparse transition dynamics.
However, they are not as effective in the dense regime, and
their irregular memory access patterns make it challenging to map
them efficiently to modern hardware~accelerators.
\item
A recent development is the use of probabilistic inference
for probabilistic model checking.
Rubicon~\cite{holtzen2021rubicon} compiles a DTMC into
the Dice~\cite{holtzen2020dice} probabilistic programming language;
running the Dice program using Dice's inference engine---weighted model counting on binary decision diagrams (BDDs)---computes
the step-bounded reachability probability.
Still, BDD operations are pointer-rich, so they involve indirection
that does not map well to devices like GPUs.
\end{itemize}

We present \tessa, a new methodology for verifying step-bounded reachability
properties of DTMCs, by compiling DTMC models to dense tensor
computations (i.e., array programs in an array programming language
like JAX \cite{jax2018github}).

As \cref{fig:motivating-example-scaling-plot} shows,
\tessa exhibits markedly better scalability than existing tools on
the $N$-professor model.
The two plots show running times for different values of $N$
(fixing the horizon $H\,{=}\,10$) and different values of $H$
(fixing $N\,{=}\,12$), respectively.
The first plot shows that,
while \tessa does not change the asymptotic complexity of the
problem (the state-space size still grows exponentially with
$N$), it aggressively pushes down the constant factor of the
exponential growth compared to \storm and \dice.
The second plot shows that while all tools appear to scale linearly
with $H$, \tessa has a significantly smaller slope---the curve is nearly flat.\footnote{The comparison may make state-of-the-art tools look
  inefficient, but they are actually highly optimized for CPUs.
The point is that a methodological shift unlocks massive speedups.
}

A key reason for \tessa's scalability is its representation of
state distributions as dense tensors, and state-transition dynamics
as transformations on these tensors.
The name \tessa evokes a tesseract---a high-dimensional cube---reflecting
our insight: rather than eagerly materializing the state-transition
dynamics as a sparse matrix or a decision diagram,
we treat the joint distribution of $N$ state variables as an order-$N$
tensor, and the transition dynamics as a tensor program.

A key advantage of this tensor representation is its amenability
to optimization (by modern tensor compilers) and
acceleration (by hardware such as GPUs).
Because \tessa maps the verification problem to tensor operations,
we can use an off-the-shelf tensor compiler (XLA~\cite{xla-github}, in this
case) to optimize the generated tensor code\footnote{\tessa running times in \cref{fig:motivating-example-scaling-plot}
  include compilation times of JAX and XLA.}
and offload them to GPU devices for massively parallel execution.
In contrast, established methodologies target representations that induce irregular memory accesses,
which limit both the degree of parallelism and
the ease of parallelization achievable.

Targeting tensors also naturally supports parameter
search for DTMCs with unknown transition probabilities.
As \tessa generates differentiable tensor programs,
we can leverage JAX's automatic differentiation to compute
gradients of a rich selection of optimization objectives.
Furthermore, the
computationally intensive gradient calculation can be optimized by XLA
and accelerated by GPUs.

We believe the methodology embodied by \tessa usefully complements
the toolbox of probabilistic model checkers.
It is not a silver bullet; existing methods
are already effective for sparse models.
However, \tessa opens up a new
regime of tractable models that were previously out of reach.

\paragraph{Contributions.}
Our contributions are both theoretical and practical:
\begin{itemize}[align=left,labelsep=*,leftmargin=*,parsep=2pt,itemsep=2pt,topsep=2pt,]
\item
We base our development on a core language for specifying DTMCs
(\cref{tessa:sec:lang}).
We show how to interpret programs in this core language as computations over tensors
(\cref{sec:tensor-sem}).
This interpretation captures the essence of the implementation of \tessa. 
We prove that this tensor interpretation is sound with respect to
the verification of step-bounded reachability properties
(\cref{sec:correctness}).

\item
We implement \tessa and evaluate it on selected benchmarks from the literature.
\tessa employs a domain-specific compiler stack that automatically
optimizes and parallelizes dense tensor computations for GPU execution
(\cref{sec:accelerate}).
Experimental results highlight substantial performance gains
unlocked by \tessa's tensor-based approach
(\cref{tessa:sec:eval}).
\end{itemize}

\section{Preliminaries}
\label{sec:preliminaries}

We review discrete-time Markov chains (DTMCs) and
the verification problem of step-bounded reachability properties.
A \emph{Markov chain} is a tuple
$\mathcal{M} = \parens*{\mathcal{S}, \iota, \eta, \mathcal{G}}$:
\begin{itemize}[align=left,labelsep=*,leftmargin=*,parsep=0pt,itemsep=0pt,topsep=2pt,]
\item 
$\mathcal{S}$ is a finite set of states.
\item 
$\iota \in \Distr{\mathcal{S}}$ is the initial-state distribution.
\item 
$\eta : \mathcal{S} \to \Distr{\mathcal{S}}$ is the transition function.
\item 
$\mathcal{G} \subseteq \mathcal{S}$ is the set of goal states.
\end{itemize}
We use $\Distr{\mathcal{S}}$ to denote the set of probability
distributions over $\mathcal{S}$.
We write $\iota(s)$ to denote the probability of starting
in state $s$.
We write $\eta\parens*{s,s'}$ to denote the probability of transitioning
to state $s'$ from state $s$.
Thus, we have that $\sum_{s \in \mathcal{S}} \iota(s) = 1$
and also that for any state $s \in \mathcal{S}$,
$\sum_{s' \in \mathcal{S}} \eta(s, s') = 1$.

In this paper, we are interested in the verification of
step-bounded reachability properties of the form
$\Pr_{\mathcal{M}}\parens*{\Diamond^{\leq n} \mathcal{G}}$, which denotes the
probability of reaching a goal state within $n$ steps.

The probability of reaching a goal state within $n$ steps
starting from state $s$ is denoted
$\Pr_{\mathcal{M}}\parens*{s \vDash \Diamond^{\leq n} \mathcal{G}}$
and is defined inductively as follows:
\begin{align}
\label{eq:reachability-via-state-base-case}
\Pr_{\mathcal{M}}\parens*{s \vDash \Diamond^{\leq 0} \mathcal{G}} & \defeq
\begin{cases}
1 & \text{if } s \in \mathcal{G} \\
0 & \text{otherwise}
\end{cases}
\\
\label{eq:reachability-via-state-inductive-case}
\Pr_{\mathcal{M}}\parens*{s \vDash \Diamond^{\leq n+1} \mathcal{G}} & \defeq
\begin{cases}
1 & \text{if } s \in \mathcal{G} \\
\sum_{s' \in \mathcal{S}} \eta(s, s') \cdot \Pr_{\mathcal{M}}\parens*{s' \vDash \Diamond^{\leq n} \mathcal{G}} & \text{otherwise}
\end{cases}
\end{align}
The inductive case \cref{eq:reachability-via-state-inductive-case}
is intuitive:
the probability of reaching a goal state within $n+1$ steps
is $1$ if we are already in a goal state;
otherwise, it is the weighted sum of the probabilities of reaching
a goal state within $n$ steps from each possible next state.

Then $\Pr_{\mathcal{M}}\parens*{\Diamond^{\leq n} \mathcal{G}}$,
the probability of reaching a goal state within $n$ steps
when starting from a state drawn from the initial-state distribution,
is defined as
\begin{gather}
\Pr_{\mathcal{M}}\parens*{\Diamond^{\leq n} \mathcal{G}} \defeq
\sum_{s \in \mathcal{S}} \iota(s) \cdot \Pr_{\mathcal{M}}\parens*{s \vDash \Diamond^{\leq n} \mathcal{G}}.
\end{gather}

We introduce another (perhaps less common) way to calculate
step-bounded reachability probabilities.
We write
$\Pr_{\mathcal{M}}\parens*{\mu\Vdash\Diamond^{\leq n} \mathcal{G}}$
to denote the probability of reaching a goal state within $n$ steps
when starting from a state drawn from distribution $\mu \in
\Distr{\mathcal{S}}$.
It is defined inductively as follows:
\begin{align}
\label{eq:reachability-via-distribution-base-case}
\Pr_{\mathcal{M}}\parens*{\mu\Vdash\Diamond^{\leq 0} \mathcal{G}} & \defeq
\sum_{s \in \mathcal{G}} \mu(s)
\\
\label{eq:reachability-via-distribution-inductive-case}
\Pr_{\mathcal{M}}\parens*{\mu\Vdash\Diamond^{\leq n+1} \mathcal{G}} & \defeq
\parens*{\sum_{s \in \mathcal{G}} \mu(s)} + \Pr_{\mathcal{M}}\parens*{\sum_{s \notin \mathcal{G}} \mu(s) \eta(s) \Vdash \Diamond^{\leq n} \mathcal{G}}
\end{align}

\noindent
The inductive case \cref{eq:reachability-via-distribution-inductive-case}
is intuitive:
the probability of reaching a goal state within $n+1$ steps
is the sum of
(1) the probability of being in a goal state at the start,
and (2) the probability of reaching a goal state within $n$ steps
after taking one step from a non-goal state.
Notice that in the second term, the next-state distribution
is given by
$\sum_{s \notin \mathcal{G}} \mu(s) \eta(s)$,
which is the weighted sum of the transition distributions
from all non-goal states.

The correctness of this alternative way to calculate step-bounded
reachability probabilities is established by
\cref{thm:reachability-via-distribution-equiv}
(and proven in the accompanying technical report~\cite{tessa-cav2026-tr}):
\begin{theorem}
\label{thm:reachability-via-distribution-equiv}
$
\Pr_{\mathcal{M}}\parens*{\Diamond^{\leq n} \mathcal{G}}
=
\Pr_{\mathcal{M}}\parens*{\iota \Vdash \Diamond^{\leq n} \mathcal{G}}
$.
\end{theorem}

\noindent
Its connection to tensor probabilistic model checking
will become clear later, but for now we simply note that
we use $\Pr_{\mathcal{M}}\parens*{\mu \Vdash \Diamond^{\leq n} \mathcal{G}}$
as a bridge to connect tensor-based reachability to
the standard definition of reachability probabilities.

\noindent

\section{A Core Model-Specification Language}
\label{tessa:sec:lang}

Directly defining a DTMC via the four-tuple formalism in
\cref{sec:preliminaries} is intractable for complex systems. The state
space typically grows exponentially with the number of variables,
making such brute-force specification error-prone and difficult to maintain.
As a result, probabilistic model checkers such as PRISM~\cite{prism4} and Storm~\cite{storm2022}
adopt high-level modeling languages that allow users to specify the
system in a modular fashion.

In this section, we define a high-level, modular specification
language that captures the core aspects of a DTMC modeling language
(e.g., the \texttt{dtmc} dialect of the PRISM language \cite{prism-manual-v4.9}).
This allows describing a system concisely as a collection of
interacting modules rather than a monolithic state-transition matrix.

We call this core language \Lang.
To place its tensor interpretation on a formal footing
(\cref{sec:tensor-sem,sec:correctness}),
we first define \Lang's syntax in \cref{tessa:sec:lang-syntax}
and its Markov-chain semantics in \cref{sec:lang-semantics}.

\subsection{Syntax of \Lang}
\label{tessa:sec:lang-syntax}

\cref{tessa:fig:syntax} presents the syntax of \Lang.

\begin{figure}
\centering
\iftight
\sm{1.5ex}
\fi
\begin{gather*}
\begin{array}{@{}r@{\ \ \ }c@{\ \ }l@{\ \ }l@{}}
  \LBL{Model} & M & \Coloneqq &
  \Key{model } m_1 ;\, \cdots ;\, m_K \Key{ goal } e
  \\
  \LBL{Module} & m & \Coloneqq &
  \Key{module } D\ C
  \\
  \LBL{Declarations} & D & \Coloneqq &
  d_1 ;\, \cdots ;\, d_n
  \\
  \LBL{Declaration} & d & \Coloneqq &
  \Key{var } x : \Key{[} 0 \; \Key{..} \; n_{\text{max}} \Key{]} \Key{ init } n_{\text{init}}
  \\
  \LBL{Commands} & C & \Coloneqq &
  c_1 ;\; \cdots ;\; c_n
  \\
  \LBL{Command} & c & \Coloneqq &
  \Key{[} a \Key{]} \;
  g \to U
  \\
  \LBL{Guard} & g & \Coloneqq &
  e
  \\
  \LBL{Mixture of updates} & U & \Coloneqq &
  \theta_1 : u_1 + \cdots + \theta_n : u_n
  \\
  \LBL{Update} & u & \Coloneqq &
  x_1 \GETS e_1 \; \cdots \; x_n \GETS e_n
  \\
  \LBL{Expression} & e & \Coloneqq &
  n \mid x \mid \textit{op}\parens*{e_1, \ldots, e_n} 
\end{array}
\\
n \in \mathbb{N} \quad
x \in \text{Variables} \quad
\theta \in [0, 1]
\end{gather*}
\caption{\Lang syntax.}
\label{tessa:fig:syntax}
\end{figure}

A model $M$ contains $K$ modules.
Each module $m_k$ ($1 \leq k \leq K$)
declares a disjoint set of state variables~$X_k$
via a set of declarations $D_k$.
The module further specifies its behavior via
a set of commands $C_k$ that update the state variables in $X_k$.
While each module declares and updates only its own state variables,
it can read the variables of all the modules in $M$.

A declaration $d$ defines a state variable $x$,
its domain (a bounded range of non-negative integers),
and its initial value.

A command $c$ consists of an action label $a$,
a guard $g$, and a mixture~$U$ of updates.
The action label allows commands from different modules to be
synchronized.
The guard $g$ is a Boolean expression specifying a necessary
(though not sufficient) condition for the command to be enabled.

A command may also be unlabeled (written $\Key{[}\,\Key{]}$ in
PRISM), in which case it is assigned a globally
unique action label.
As a simplification, we assume that all unlabeled commands have
already been assigned globally unique action labels.

Each update $u_i$ in $U$ is associated with a probability $\theta_i$,
where $\sum_i \theta_i = 1$.
Each update further consists of deterministic assignments to
a subset of the module's declared state variables.

An expression $e$ is either a constant $n$,
a state variable $x$,
or an $n$-ary operation over sub-expressions $e_1, \ldots, e_n$.
As a simplification, we treat Boolean expressions
as integer expressions where $0$ represents false
and non-zero values represent true.

The global model $M$ also specifies the goal states
via a Boolean expression $e$ over the state variables of all modules.

\subsection{Semantics of \Lang}
\label{sec:lang-semantics}

We now define the standard semantics of \Lang
by interpreting a model $M$ as a DTMC
$\mathcal{M}=\parens*{\textit{State}\Bracks*{M}, \textit{Init}\Bracks*{M}, \textit{Step}\sem{M}, \textit{Goal}\Bracks*{M}}$.

Let $X_k$ be the set of state variables declared in module $m_k$ ($1 \leq k \leq K$).
Let \vcrunch{$\Vars{M} \defeq \biguplus_{k=1}^{K} X_k$} be the set of all state
variables in $M$.
Let $\dom{x}$ denote the domain of variable $x$.
For $X \subseteq \Vars{M}$,
the state space formed by the variables in $X$ is
\begin{gather}
\STATE\parens*{X} \defeq
\parens*{x \in X} \to \dom{x}.
\end{gather}

\noindent
That is, a state is a mapping from each variable in $X$
to a value in its domain.
The state space of the entire model $M$ is then defined as
\begin{gather}
\textit{State}\Bracks*{M} \defeq
\STATE\parens*{\Vars{M}}.
\end{gather}
We will write $\STATE$ as a shorthand for $\textit{State}\Bracks*{M}$
when $M$ is clear from context.

Notice that $\STATE\parens*{X}$ is measurable since it is a finite set.
The size of $\STATE\parens*{X}$ is
$
\verts*{\STATE\parens*{X}} =
\prod_{x \in X} \verts*{\dom{x}}
$,
which grows exponentially with the number of state variables in~$X$.
We write $s(x)$ to denote the value of variable $x$ in state $s$.

Let $s_0 \in \STATE$ be the initial state
where each variable $x \in \Vars{M}$
is initialized to the value specified
in its declaration.
Then the initial-state distribution~is
\begin{gather}
\textit{Init}\Bracks*{M} \defeq \delta_{s_0}
\end{gather}
where $\delta_{s}$ is the delta distribution (point mass) at $s$.

Let $e$ be the goal expression of $M$.
Then the set of goal states is defined as
\begin{gather}
\textit{Goal}\Bracks*{M} \defeq
\setc{s \in \STATE}{\sem{e} (s) \neq 0}
\end{gather}
where $\sem{e} \, s$ is the interpretation of expression $e$
under state $s$.

All that remains is to define $\sem{e}$ and $\textit{Step}\sem{M}$.
To define $\textit{Step}\sem{M}$, we need to first define the semantics of
each module $m_k$ in $M$, which further requires defining the semantics
of each command in $m_k$ and each update in a command.

\paragraph{Interpreting expressions.}
$\sem{e} : \STATE\parens*{X} \to \mathbb{N}$,
where $X$ is a superset of the variables appearing in $e$,
is defined inductively as follows:
\begin{mathparpagebreakable}
\Rule{}{
  \sem{n} \, s \defeq n
}

\Rule{}{
  \sem{x} \, s \defeq s(x)
}

\Rule{
  \sem{e_i} \, s = n_i \text{ for } i = 1, \ldots, l
  \\
  \sem{\textit{op}} = f
}{
  \sem{\textit{op}\parens*{e_1, \cdots, e_l}} \, s \defeq f(n_1, \ldots, n_l)
}
\end{mathparpagebreakable}

\noindent

\paragraph{Interpreting updates.}
Recall that $X_k$ denotes the set of state variables declared in $m_k$.
Let $Y_k$ be the set of variables each of which is either
declared in $m_k$ or read in $m_k$'s commands.
So we have that $X_k \subseteq Y_k \subseteq \Vars{M}$.

An update $u$ in module $m_k$ is interpreted as
a function $\sem{u} : \STATE\parens*{Y_k} \to \STATE\parens*{X_k}$.
That is,
for $s \in \STATE\parens*{Y_k}$,
\begin{gather}
\label{eq:update-interp}
\sem{x_1 \GETS e_1 \; \cdots \; x_n \GETS e_n} \, s \defeq
{ \restrict{s}{X_k} }
\bracks*{x_1 \mapsto \sem{e_1}s, \; \cdots, \; x_n \mapsto \sem{e_n}s}.
\end{gather}
In \cref{eq:update-interp},
the notation $\restrict{s}{X_k}$ denotes the restriction of $s$ to the
variables in $X_k$,
and $\restrict{s}{X_k}\bracks*{x_1 \mapsto \sem{e_1}s, \; \cdots, \; x_n \mapsto \sem{e_n}s}$
denotes the state obtained by updating $\restrict{s}{X_k}$ with the
assignments in $u$.

\paragraph{Interpreting mixtures of updates.}
A mixture of updates $U$ in module $m_k$ is interpreted as a function $\sem{U} : \STATE\parens*{Y_k} \to \Distr{\STATE\parens*{X_k}}$.
That is,
\begin{gather}
\label{eq:mixture-interp}
\sem{\theta_1 : u_1 + \cdots + \theta_n : u_n} \, s \defeq
\sum_{i=1}^{n} \theta_i \cdot \delta_{\sem{u_i} s}.
\end{gather}

\paragraph{Interpreting modules.}
Let $\ACTION\parens*{M}$ be the model $M$'s alphabet of action labels.
Let $\ACTION\parens*{m_k}$ be the set of action labels appearing in $m_k$'s commands.
The module $m_k$ is interpreted as a function
$\sem{m_k} : \ACTION\parens*{M} \to \STATE\parens*{Y_k} \to \UNDistr{\STATE\parens*{X_k}}$.
Here, $\UNDistr{\STATE\parens*{X_k}}$ is the set of unnormalized distributions
over $\STATE\parens*{X_k}$;
these distributions have finite total mass, but the total mass is not necessarily~$1$.
Specifically,
\begin{gather}
\label{eq:module-interp}
\sem{m_k} \, a \, s \defeq
\begin{cases}
\delta_{\restrict{s}{X_k}} & \text{if $a \not\in \ACTION\parens*{m_k}$},
\\
\sum_{\Key{[} a \Key{]} \, g \,\to\, U \;\in\; \commands{m_k, a}}
\indicator{\sem{g} \, s \neq 0} \cdot \sem{U} \, s
& \text{otherwise.}
\end{cases}
\end{gather}

\noindent
In \cref{eq:module-interp},
$\commands{m_k, a}$ is the set of commands in module $m_k$ with action label $a$,
and $\indicator{P}$ is the indicator function that evaluates to $1$ if
predicate $P$ is true and $0$ otherwise.

Intuitively, $\sem{m_k} \, a \, s$ describes the aggregate transition
behavior of module~$m_k$ when the current state is $s$ and
the chosen action is $a$.
There are two cases.
If $a$ is not in $m_k$'s alphabet, then $m_k$ stays put.
Otherwise, a sum ranges over all commands with action label~$a$ whose
guard is satisfied by~$s$.
When exactly one such command
$\Key{[} a \Key{]} \, g \,\to\, U$
exists, $\sem{m_k} \, a \, s$ is the distribution
denoted by $\sem{U} \, s$.
When multiple commands with the action label~$a$ have overlapping guards,
$\sem{m_k} \, a \, s$ is an unnormalized distribution whose total mass
equals the number of enabled commands;
PRISM allows this ambiguity and resolves it by a random choice
among all the enabled command combinations (see \cref{eq:step-interp}).
If no command with action label~$a$ is enabled in~$s$,
$\sem{m_k} \, a \, s$ has total mass~$0$.

\paragraph{Interpreting models.}

In a specification language like PRISM \cite{prism-manual-v4.9},
the global behavior of a DTMC model is defined by the parallel
execution of modules synchronized by action labels.
A global transition step corresponds to the simultaneous transitioning
of all modules according to a common action label.

Let $\textit{En}\sem{m_k} \, a \, s$,
where $s \in \STATE\parens*{Y_k}$,
count the number of commands with action label~$a$ that are locally enabled
in module~$m_k$ in state~$s$:
\begin{align}
\label{eq:enabled-local-interp}
\textit{En}\sem{m_k} \, a \, s &\defeq 
\begin{cases}
1 & \text{if } a \not\in \ACTION(m_k), \\ 
\sum_{\Key{[} a \Key{]} \, g \,\to\, U \;\in\; \commands{m_k, a}} \indicator{\sem{g} \, s \neq 0} & \text{otherwise,}
\end{cases}
\end{align}
In case $a \notin \ACTION\parens*{m_k}$,
it is considered that there is one enabled command---the implicit self-loop---so $\textit{En}\sem{m_k} \, a \, s$ is defined to be~$1$.

Let $\textit{En}\sem{M} \, a \, s$,
where $s \in \STATE$,
count the total number of enabled \emph{command combinations}
for action~$a$ across all modules:
\begin{align}
\label{eq:enabled-global-interp}
\textit{En}\sem{M} \, a \, s &\defeq
\prod_{k=1}^{K} \textit{En}\sem{m_k} \, a \, {\restrict{s}{Y_k}}.
\end{align}

\noindent
Now we can define the DTMC transition function
$\textit{Step}\sem{M} : \STATE \to \Distr{\STATE}$:
\begin{gather}
\label{eq:step-interp}
\textit{Step}\Bracks*{M} \, s \defeq 
\begin{cases} 
\displaystyle
\frac{\sum_{a}\bigotimes_{k=1}^{K} \sem{m_k} \, a \, {\restrict{s}{Y_k}}}{\sum_{a} \textit{En}\sem{M} \, a \, s}
&
\displaystyle
\text{if $\sum_{a} \textit{En}\sem{M} \, a \, s > 0$,}\\
\delta_s & \text{otherwise.}
\end{cases}
\end{gather}

\noindent
The cases in \cref{eq:step-interp} depend on
whether there is at least one enabled command combination in state~$s$.
\begin{itemize}[align=left,labelsep=*,leftmargin=*,parsep=2pt,itemsep=2pt,topsep=2pt,]
\item
In PRISM, multiple commands---possibly with different action labels,
or with the same label but overlapping guards---may be simultaneously
enabled in a state for a module.
This ambiguity is resolved by a uniformly random choice over all enabled
command combinations across all modules.

A command combination for action~$a$ consists of one enabled command from
each module; the number of such command combinations is
{$\textit{En}\sem{M} \, a \, s = \prod_{k} \textit{En}\sem{m_k} \, a \, {\restrict{s}{Y_k}}$}.
The denominator $\sum_{a} \textit{En}\sem{M} \, a \, s$ in \cref{eq:step-interp}
is the total number of enabled command combinations across all actions.

The operator $\bigotimes$ in \cref{eq:step-interp}
denotes the product of (unnormalized) distributions over
disjoint sets of variables.
Since the $K$ disjoint sets of variables $X_1, \ldots, X_K$
aggregate to $\Vars{M}$,
\vcrunch{$\bigotimes_{k=1}^{K} \sem{m_k} \, a \, {\restrict{s}{Y_k}}$}
is an (unnormalized) distribution over the global state space $\STATE$.
Moreover, its total mass equals $\textit{En}\sem{M} \, a \, s$.

\item
If no action is enabled in $s$, the DTMC loops back to $s$.
\end{itemize}

While
$\bigotimes_{k=1}^{K} \sem{m_k} \, a \, {\restrict{s}{Y_k}}$
is in general an unnormalized probability
distribution, $\textit{Step}\sem{M} \, s$ is a probability distribution.
We prove this fact (\cref{thm:model-distribution-preservation})
in the technical report~\cite{tessa-cav2026-tr}.

\begin{theorem}
\label{thm:model-distribution-preservation}
For any $s \in {\STATE}$, $\textit{Step}\sem{M} \, s \in \Distr{\STATE}$.
\end{theorem}

\section{Probabilistic Model Checking as Tensor Computations}
\label{sec:tensor-sem}

In this section, we show how to cast the verification of \Lang
models as computations over tensors.
\cref{sec:tensor-rep} maps state distributions to tensors.
\cref{sec:tensor-transformer-sem} maps \Lang models to tensor
transformers.
\cref{sec:tensor-pmc} maps the verification of step-bounded
reachability properties to tensor computations.

\subsection{Representing State Distributions as Tensors}
\label{sec:tensor-rep}

We use tensors to represent discrete probability distributions over
DTMC states.
A tensor is a multidimensional array generalizing scalars (order-0),
vectors (order-1), matrices (order-2), and so on.
In this work, we use tensors
$\vec{T} \in \mathbb{R}^{a_1 \times \cdots \times a_N}$
over the field of real numbers $\mathbb{R}$
and tensors
$\vec{T} \in \mathbb{N}^{a_1 \times \cdots \times a_N}$
over the field of natural numbers $\mathbb{N}$,
where $N$ is
the \emph{order} (a.k.a.\ \emph{rank}) of the tensor and $a_k$ is the
size of the $k$-th dimension.

\paragraph{The state space as an index space for tensors.}
A state space
$
\STATE\parens*{X}=\parens*{x \in X} \to \dom{x}
$
serves as the index space for order-$\verts*{X}$ tensors:
\begin{itemize}[align=left,labelsep=*,leftmargin=*,parsep=0pt,itemsep=0pt,topsep=2pt,]
\item each dimension corresponds to a state variable $x \in X$,
\item and the size of each dimension is $\verts*{\dom{x}}$.
\end{itemize}

We write $\Tensor{\STATE\parens*{X}}$ to denote the set of
$\mathbb{R}$-valued
order-$\verts*{X}$ tensors with the index space $\STATE\parens*{X}$.
That is,
{$
\Tensor{\STATE\parens*{X}} \defeq
\mathbb{R}^{\prod_{x \in X} \verts*{\dom{x}}}
$}.
For $\vec{T} \in \Tensor{\STATE\parens*{X}}$,
$\tindex{\vec{T}}{s}$ denotes the entry of $\vec{T}$
at index $s \in \STATE\parens*{X}$.

$\Tensor{\STATE\parens*{X}}$ is a \emph{vector space} (a.k.a.\ \emph{linear space})
under element-wise addition and scalar multiplication.
This $\verts*{X}$-dimensional vector space has a natural set of basis vectors
$\setc{\vec{B}_s}{s \in \STATE\parens*{X}}$,
where each basis vector $\vec{B}_s \in \Tensor{\STATE\parens*{X}}$
is a tensor defined as
$
\tindex{\vec{B}_s}{s'} \defeq
\indicator{s' = s}
$.

\paragraph{State distributions as tensors.}
We define a function ${\Xi} : \UNDistr{\STATE\parens*{X}} \to \Tensor{\STATE\parens*{X}}$ mapping
a distribution $\mu \in \UNDistr{\STATE\parens*{X}}$ to a tensor:
\[
{\Xi}(\mu) \defeq \sum_{s \in \STATE\parens*{X}} \mu(s) \cdot \vec{B}_s.
\]
That is,
${\Xi}(\mu)\,{=}\,\vec{T}$ iff
$\forall s\,{\in}\,\STATE\parens*{X}$, $\tindex{\vec{T}}{s} = \mu(s)$.
For example, given a joint distribution~$\mu$ of two state variables
(say, the states of professors in \cref{fig:motivating-example} when $N=2$),
$\Xi$ maps $\mu$ to an order-2 tensor $\vec{T} \in \mathbb{R}^{3 \times 3}$:

\begin{centered}
\sm{1ex}
\def\MyFont{\fontsize{7}{8.5}\selectfont}
\begin{tikzpicture}
\node[align=left, font=\MyFont, fill=gray!20, rounded corners=2pt, inner sep=2pt] (dist-text) at (0,0) {\begin{minipage}{42pt}\centering
        distribution\\$\mu$
        \end{minipage}
    };

\node[align=left, font=\MyFont] (dist) [right=3pt of dist-text] {\def\arraystretch{1.0}$\begin{array}{@{\;}c@{\ \ \ }c@{\ \ \ }c@{\;}}
            s_1 & s_2 & \text{probability} \\
            \hline
            0 & 0 & 0.1 \\
            0 & 1 & 0.2 \\
            0 & 2 & 0.0 \\
            1 & 0 & 0.3 \\
\cdots & \cdots & \cdots
        \end{array}$
    };

\matrix[tensorgrid, font=\MyFont] (tensor) [right=80pt of dist] {
        0.1 & 0.2 & 0.0 \\
        0.3 & 0.0 & 0.1 \\
        0.0 & 0.1 & 0.2 \\
    };

\node[tensorlabel, left=0pt of tensor-1-1] {$s_1{=}0$};
    \node[tensorlabel, left=0pt of tensor-2-1] (s1-lbl) {$s_1{=}1$};
    \node[tensorlabel, left=0pt of tensor-3-1] {$s_1{=}2$};
    
\node[tensorlabel, above=10pt of tensor-1-1, rotate=45, anchor=center] {$s_2{=}0$};
    \node[tensorlabel, above=10pt of tensor-1-2, rotate=45, anchor=center] {$s_2{=}1$};
    \node[tensorlabel, above=10pt of tensor-1-3, rotate=45, anchor=center] {$s_2{=}2$};

    \node[right=2pt of tensor, font=\MyFont, fill=gray!20, rounded corners=2pt, inner sep=2pt] (tensor-text) {\begin{minipage}{22pt}\centering
        tensor\\$\vec{T}$
        \end{minipage}
    };

\draw[arrowstyle] (dist.east) -- (s1-lbl.west) node[midway, above, font=\MyFont, fill=gray!20, rounded corners=2pt, inner sep=2pt, outer sep=3pt] {$\Xi(\mu)=\vec{T}$
    };
   
\end{tikzpicture}
 \end{centered}

\noindent
Each dimension corresponds to a professor's state.
Both dimensions have size~3 (corresponding to states \textsf{Away}, \textsf{Doodling}, \textsf{Done}).
This tensor should not be confused with the state-transition matrix
used by a probabilistic model checker like Storm~\cite{storm2022},
which would be of size $9 \times 9$ (since there are 9 joint states).

\paragraph{Index tensors.}

For each variable $x \in X$, we define its \emph{index tensor}
$\vec{I}_{x \in X} \in \Tensor[N]{\STATE\parens*{X}}$
as a $\mathbb{N}$-valued order-$|X|$ tensor.
For any index $s \in \STATE\parens*{X}$,
$\tindex{\vec{I}_{x \in X}}{s} = s\parens*{x}$
is simply the value of that variable in state $s$.

\subsection{Tensor Semantics of \Lang Models}
\label{sec:tensor-transformer-sem}

\paragraph{Interpreting expressions.}
\cref{sec:lang-semantics} interprets an expression $e$ as a function
$\sem{e} : \STATE\parens*{X} \to \mathbb{N}$.
We now define $\tsem{e}_X \in \Tensor[N]{\STATE\parens*{X}}$ as follows,
where the subscript~$X$ indicates the shape of the tensor:
\begin{mathparpagebreakable}
\Rule{}{
  \tsem{n}_X \defeq n \cdot \mathbf{1}_{X}
}

\Rule{}{
  \tsem{x}_X \defeq \vec{I}_{x \in X}
}

\Rule{
  \tsem{e_i}_X = \vec{T_i} \text{ for } i = 1, \ldots, l
  \\
  \tsem{\textit{op}}_X = f
}{
  \tsem{\textit{op}\parens*{e_1, \cdots, e_l}}_X \defeq f(\vec{T_1}, \ldots, \vec{T_l})
}
\end{mathparpagebreakable}

\noindent
where $\mathbf{1}_{X} \in \Tensor[N]{\STATE\parens*{X}}$ is the all-ones tensor,
and $\tsem{\textit{op}}_X$ is $\sem{\textit{op}}$ lifted to operate
element-wise over tensors.

\paragraph{Interpreting updates.}
\cref{sec:lang-semantics} interprets an update $u$ in module $m_k$ as a function $\sem{u} : \STATE\parens*{Y_k} \to \STATE\parens*{X_k}$.
So given $s \in \STATE\parens*{Y_k}$, $\sem{u} \, s$ gives the value
of $x \in X_k$ after executing $u$ in state $s$.
This reading suggests that we can define the tensor interpretation of $u$ as a function $\tsem{u}_{Y_k} : X_k \to \Tensor[N]{\STATE\parens*{Y_k}}$:
\begin{gather}
\label{eq:update-tensor}
\tsem{
  x_1 \GETS e_1 \; \cdots \; x_n \GETS e_n
}_{Y_k} \, x \defeq
\begin{cases}
  \tsem{e_i}_{Y_k} & \text{if $x = x_i$ for some $i$,}
  \\
  \vec{I}_{x \in Y_k} & \text{otherwise.}
\end{cases}
\end{gather}

\paragraph{Interpreting mixtures of updates.}
\cref{sec:lang-semantics} interprets a mixture of updates~$U$ in module~$m_k$
as a function $\sem{U} : \STATE\parens*{Y_k} \to \Distr{\STATE\parens*{X_k}}$.
We now interpret $U$ as a tensor
$\tsem{U} \in \Tensor{\STATE\parens*{Y_k \uplus X_k}}$,
where $\uplus$ is the disjoint-union operator.
The tensor $\tsem{U}$ is of order $\verts*{Y_k} + \verts*{X_k}$.
Specifically,
\begin{gather}
\label{eq:mixture-tensor}
\tsem{\theta_1 : u_1 + \cdots + \theta_n : u_n} \defeq \sum_{i=1}^{n} \theta_i \cdot \bigodot_{x \in X_k} \indicator{
  \parens*{\tsem{u_i}_{Y_k} \, x} \otimes \mathbf{1}_{X_k} =
  \mathbf{1}_{Y_k} \otimes \vec{I}_{x \in X_k}
},
\end{gather}
where $\bigodot$ is the Hadamard (i.e., element-wise) product,
$\otimes$ is the outer product,
and $=$ is overloaded to denote element-wise equality.
Observe the correspondence between \cref{eq:mixture-interp} and \cref{eq:mixture-tensor}.
The outer products in \cref{eq:mixture-tensor} are used to align the
tensor dimensions.
Intuitively, if we write $\parens*{s, s'}$ as an index,
where $s \in \STATE\parens*{Y_k}$ and $s' \in \STATE\parens*{X_k}$,
then we should have
$
\tindex{\tsem{\theta_1 : u_1 + \cdots + \theta_n : u_n}}{\parens*{s, s'}} =
\sum_{i=1}^{n} \theta_i \cdot \indicator{\sem{u_i} s = s'}
$.
That is, the tensor's value at index $\parens*{s, s'}$ is the probability
of transitioning to state $s'$ from state $s$ when the mixture of
updates $U$ is executed.

\paragraph{Interpreting modules.}
\cref{sec:lang-semantics} interprets a module $m_k$ as a function
$\sem{m_k} : \ACTION\parens*{M} \to \STATE\parens*{Y_k} \to \UNDistr{\STATE\parens*{X_k}}$.
Accordingly, we define the tensor interpretation
$\tsem{m_k} : \ACTION\parens*{M} \to \Tensor{\STATE\parens*{Y_k \uplus X_k}}$
as follows:
\begin{align}
\label{eq:module-tensor}
\tsem{m_k} \, a \defeq
\begin{cases} 
\bigodot_{x \in X_k} \indicator{ \vec{I}_{x \in Y_k} \otimes \mathbf{1}_{X_k} = \mathbf{1}_{Y_k} \otimes \vec{I}_{x \in X_k} } & \text{\!if } a \not\in \ACTION(m_k), \\
\sum_{\Key{[} a \Key{]} g \to U \,\in\, \commands{m_k, a}} \parens*{ \indicator{\tsem{g}_{Y_k} \neq 0} \otimes \mathbf{1}_{X_k} } \odot \tsem{U} & \text{\!otherwise.}
\end{cases}
\end{align}
Observe the correspondence between \cref{eq:module-interp} and \cref{eq:module-tensor}.
Intuitively, $\tsem{m_k} \, a$ is a tensor whose value at index $\parens*{s, s'}$
is the unnormalized probability of transitioning to state $s'$ from state $s$ when
$m_k$ is executed with action $a$.
When multiple commands have overlapping guards,
the sum $\sum_{s'} \tindex{\tsem{m_k} \, a}{(s,s')}$
equals the number of enabled commands in state~$s$;
normalization to a probability distribution occurs
at the model level rather than at the module level
(see \cref{eq:step-tensor-pointful}).

\paragraph{Interpreting models.}

\cref{sec:lang-semantics} interprets a model $M$ as a function
$\textit{Step}\sem{M} : {\STATE} \to \Distr{\STATE}$.
Following the way that modules are interpreted as tensors,
we could interpret $M$ as a tensor
$\tsem{M} : \Tensor{\STATE\parens*{\Vars{M} \uplus \Vars{M}}}$.
However, this tensor would be of order $2\verts*{\Vars{M}}$,
effectively materializing the full state-transition matrix.
For space efficiency, we instead interpret $M$ as a function
$\tsem{M} : \Tensor{\STATE} \to \Tensor{\STATE}$,
representing the transition dynamics as a tensor transformer
(i.e., code) rather than as a tensor (i.e., data).

Corresponding to
$\textit{En}\sem{m_k} \, a \, s$ and $\textit{En}\sem{M} \, a \, s$
in \cref{eq:enabled-local-interp} and \cref{eq:enabled-global-interp},
we define tensors
$\textit{En}\tsem{m_k} \, a \in \Tensor[N]{\STATE\parens*{Y_k}}$
and $\textit{En}\tsem{M} \, a \in \Tensor[N]{\STATE}$:
\begin{align}
\textit{En}\tsem{m_k} \, a &\defeq
\begin{cases}
\mathbf{1}_{Y_k} & \text{if } a \not\in \ACTION(m_k), \\ 
\sum_{\Key{[} a \Key{]} \, g \,\to\, U \;\in\; \commands{m_k, a}} \indicator{\tsem{g}_{Y_k} \neq 0} & \text{otherwise,}
\end{cases}
\\
\tindex{\textit{En}\tsem{M} \, a}{s} &\defeq \prod_{k=1}^{K} \tindex{\textit{En}\tsem{m_k} \, a}{\restrict{s}{Y_k}}
\end{align}
Define $\vec{L} \defeq \sum_{a} \textit{En}\tsem{M} \, a$
as the tensor counting the total number of
enabled command combinations in each state.

We define $\tsem{M}$
by specifying how it transforms a tensor $\vec{T} \in \Tensor{\STATE}$ representing the current-state distribution
into a tensor representing the next-state distribution.
Specifically, for any $s' \in \STATE$,
$\tsem{M} \, \vec{T}$ weights the probability of
transitioning to $s'$ from each state $s$ by the probability
$\tindex{\vec{T}}{s}$ of being in $s$:
\begin{align}
\label{eq:step-tensor-pointful}
\tindex{\parens*{\tsem{M} \, \vec{T}}}{s'} \defeq
\sum_{s \in \STATE}
\tindex{\vec{T}}{s} \cdot 
\begin{cases}
\frac{\sum_{a}\prod_{k=1}^{K} \tindex{\tsem{m_k} \, a}{\parens*{\restrict{s}{Y_k}, \restrict{s'}{X_k}}}}{\tindex{\vec{L}}{s}}
& \displaystyle
\text{if $\tindex{\vec{L}}{s} > 0$,}\\
\indicator{s = s'} & \text{otherwise.}
\end{cases}
\end{align}

\noindent
The transition probability is
given by a case analysis, similarly to \cref{eq:step-interp}.
The branching control flow in \cref{eq:step-tensor-pointful} hinders parallelization
over the index space $\STATE$, however.
Fortunately, we can encode the branching logic as pure tensor flow,
through a sum of two terms.
We redefine $\tsem{M}$ as follows:
\begin{align}
\label{eq:step-tensor-pointfree}
\tsem{M} \, \vec{T} &\defeq
\parens*{\sum_{a} \vec{P}_a}
+
\parens*{\vec{T} \odot \indicator{\vec{L} = \mathbf{0}_{\Vars{M}}}}
\\
\label{eq:step-tensor-contraction}
\tindex{\vec{P}_a}{s'} &\defeq
\sum_{s \in \STATE}
  \tindex{\parens*{
    \frac{\vec{T}}{\max\parens*{\vec{L}, \mathbf{1}_{\Vars{M}}}}
  }}{s}
  \cdot \prod_{k=1}^{K} \tindex{\tsem{m_k} \, a}{\parens*{\restrict{s}{Y_k}, \restrict{s'}{X_k}}}
\end{align}
The two terms in \cref{eq:step-tensor-pointfree} correspond to the two branches of \cref{eq:step-tensor-pointful}.

\begin{itemize}[align=left,labelsep=*,leftmargin=*,parsep=2pt,itemsep=2pt,topsep=2pt,]
\item
The first term $\sum_a \vec{P}_a$ rearranges the first
branch of \cref{eq:step-tensor-pointful}: it pushes $\sum_{s}$ inside
$\sum_{a}$ and weights the input tensor $\vec{T}$ by
$\mathbf{1} \oslash \max\parens*{{\vec{L}}, \mathbf{1}}$,
as required by the averaging semantics.
In particular, when $\tindex{\vec{L}}{s} = 0$,
the normalization factor
$\max\parens*{\tindex{\vec{L}}{s}, {1}}=1$, while
$\tindex{\tsem{m_k} \, a}{\parens{\restrict{s}{Y_k}, \cdot}}$ is zero for some
module~$k$ for every action~$a$,
so the term vanishes.

Notice that $\vec{P}_a$ as defined in \cref{eq:step-tensor-contraction}
is a tensor contraction over dimensions $s$ corresponding
to the current state.
The output dimensions $s'$ correspond to the next state.
The tensors in this contraction are the masked input tensor and the
$K$ module tensors $\tsem{m_k} \, a$.
This tensor contraction is the main work performed by the tensor
transformer $\tsem{M}$.

\item
The second term $\vec{T} \odot \indicator{\vec{L} = \mathbf{0}}$
weights the input tensor $\vec{T}$ by the mask $\indicator{\vec{L} = \mathbf{0}}$.
This mask is $1$ exactly where no actions are enabled,
leaving the probability mass in those states unchanged
and zeroing out the mass in all other states.
\end{itemize}

\subsection{Casting Probabilistic Model Checking as Tensor Computations}
\label{sec:tensor-pmc}

With the tensor-transformer interpretation defined,
we can now cast probabilistic model checking
as tensor computations.
In words, step-bounded reachability probabilities
can be computed via repeated applications of the tensor transformer
$\tsem{M}$.
Specifically, for model $M$ with goal expression $e$,
we write $\Pr_M \parens*{\vec{T} \Vvdash \Diamond^{\leq n} e}$
to denote the probability of reaching a state satisfying expression $e$
within $n$ steps when starting from a state drawn from
the distribution represented by tensor $\vec{T} \in \Tensor{\STATE}$.
It is defined inductively as follows,
where
$\vec{\Delta}_{e} \defeq \indicator{\tsem{e}_{\Vars{M}} \neq \mathbf{0}_{\Vars{M}}}$ and
$\vec{\Delta}_{\neg e} \defeq \indicator{\tsem{e}_{\Vars{M}} = \mathbf{0}_{\Vars{M}}}$
are mask tensors:
\begin{align}
\label{eq:reachability-via-tensor-base-case}
\Pr_M \parens*{\vec{T} \Vvdash \Diamond^{\leq 0} e} & \defeq
\angles{{\vec{\Delta}_{e}}, {\vec{T}}}
\\
\label{eq:reachability-via-tensor-inductive-case}
\Pr_M \parens*{\vec{T} \Vvdash \Diamond^{\leq n+1} e} & \defeq
\angles{{\vec{\Delta}_{e}}, {\vec{T}}} +
\Pr_M \parens*{\tsem{M} \, \parens*{\vec{\Delta}_{\neg e} \odot \vec{T}} \Vvdash \Diamond^{\leq n} e}
\end{align}

The definition mirrors that of
$\Pr_{\mathcal{M}}\parens*{\mu\Vdash\Diamond^{\leq n} \mathcal{G}}$ in
\cref{eq:reachability-via-distribution-base-case,eq:reachability-via-distribution-inductive-case}.
In the base case \cref{eq:reachability-via-tensor-base-case},
the probability of being in a goal state is given by the Frobenius inner product
$\angles{{\vec{\Delta}_{e}}, {\vec{T}}}=\sum_{s \in \STATE} \tindex{\vec{\Delta}_{e}}{s} \cdot \tindex{\vec{T}}{s}$.
In the inductive case \cref{eq:reachability-via-tensor-inductive-case},
the next-state tensor is given by
$\tsem{M} \, \parens*{\vec{\Delta}_{\neg e} \odot \vec{T}}$,
which is the result of applying the tensor transformer $\tsem{M}$
to the current-state tensor masked by
$\vec{\Delta}_{\neg e}$.
The mask tensor $\vec{\Delta}_{\neg e}$
ensures that only the non-goal states in the current-state tensor
contribute to the next-state tensor.
\cref{fig:tensor-reach} illustrates \cref{eq:reachability-via-tensor-inductive-case}
for the 2-professor example.

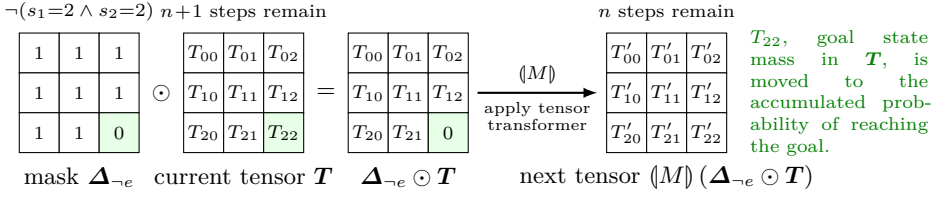
\begin{figure}[t]

\centering
\!\!\!\begin{tikzpicture}[
    flow/.style={draw, rounded corners, align=center, font=\small, fill=blue!5},
    tensorgrid/.append style={
        nodes={minimum size=15pt}
    },
]
\matrix[tensorgrid] (Mask) {
        1 & 1 & 1 \\
        1 & 1 & 1 \\
        1 & 1 & |[fill=green!10]| 0 \\
    };
    \node[below=-1pt of Mask] {mask $\vec{\Delta}_{\neg e}$};
    \node[above=-2pt of Mask, font=\scriptsize] {$\neg \parens*{s_1{=}2 \land s_2{=}2}$};

\node[right=-2pt of Mask] (odot) {$\odot$};

\matrix[tensorgrid, right=-2pt of odot] (Current) {
        T_{00} & T_{01} & T_{02} \\
        T_{10} & T_{11} & T_{12} \\
        T_{20} & T_{21} & |[fill=green!10]| T_{22} \\
    };
    \node[below=-1pt of Current] {current tensor $\vec{T}$};
    \node[above=-2pt of Current, font=\scriptsize] {$n\!+\!1$ \textls[-20]{steps remain}};

\node[right=-2pt of Current] (equal) {$=$};

\matrix[tensorgrid, right=-2pt of equal, anchor=west] (Masked) {
        T_{00} & T_{01} & T_{02} \\
        T_{10} & T_{11} & T_{12} \\
        T_{20} & T_{21} & |[fill=green!10]| 0 \\
    };
    \node[below=-1pt of Masked] {$\vec{\Delta}_{\neg e} \odot \vec{T}$};

\matrix[tensorgrid, right=45pt of Masked] (Next) {
        T'_{00} & T'_{01} & T'_{02} \\
        T'_{10} & T'_{11} & T'_{12} \\
        T'_{20} & T'_{21} & T'_{22} \\
    };
    \node[below=-2pt of Next] {next tensor $\tsem{M}\,\parens*{\vec{\Delta}_{\neg e} \odot \vec{T}}$};
    \node[above=-2pt of Next, font=\scriptsize] {$n$ \textls[-20]{steps remain}};

\draw[arrowstyle] (Masked.east) -- (Next.west) node[midway, above, font=\fontsize{7}{8}\selectfont, align=center] {$\tsem{M}$};
    \node[below=0pt of $(Masked.east)!0.5!(Next.west)$, font=\fontsize{6.0}{6.5}\selectfont, align=center] {apply tensor\\transformer};

\node[right=2pt of Next, font=\fontsize{7.0}{8.2}\selectfont, color=green!50!black, align=left, anchor=west] (goal-state-text) {\begin{minipage}{67pt}$T_{22}$, goal state mass in $\vec{T}$, is moved to the accumulated probability of reaching the goal.\end{minipage}};

\end{tikzpicture}
\caption{Visualizing the computation in \cref{eq:reachability-via-tensor-inductive-case} of the reachability probability. The probability mass $T_{22}$ at the goal state ($s_1=2, s_2=2$) is extracted and accumulated to the reachability probability. The remaining mass is transformed by the model's tensor-transformer interpretation to produce the next-state tensor.}
\label{fig:tensor-reach}
\end{figure}
 
\subsection{Discussion}

The correctness of $\Pr_M \parens*{\vec{T} \Vvdash \Diamond^{\leq n} e}$
with respect to the verification problem stated in
\cref{sec:preliminaries} will be established in \cref{sec:correctness}.
The close correspondence between
the definitions of
$\Pr_{\mathcal{M}}\parens*{\mu\Vdash\Diamond^{\leq n} \mathcal{G}}$
and $\Pr_M \parens*{\vec{T} \Vvdash \Diamond^{\leq n} e}$
makes short work of proving the
correctness of tensor probabilistic model checking.

Two factors contribute to the efficiency of
tensor probabilistic model checking.
One factor is that the computations in
\cref{eq:step-tensor-pointfree}, \cref{eq:step-tensor-contraction},
\cref{eq:reachability-via-tensor-base-case}, and
\cref{eq:reachability-via-tensor-inductive-case}
avoid materializing the full $3^N\,{\times}\,3^N$ transition matrix of
the DTMC (using the $N$-professor example for concreteness).
Instead,
the transition matrix is implicitly encoded as tensor computations
(i.e., code rather than data) that transform the order-$N$ tensors.

A more important factor contributing to the efficiency
is that the tensor computations in
\cref{eq:step-tensor-pointfree}, \cref{eq:step-tensor-contraction},
\cref{eq:reachability-via-tensor-base-case}, and
\cref{eq:reachability-via-tensor-inductive-case}
are composed of standard operations over dense tensors and, therefore,
can be implemented as first-order array programs
in an array programming language such as JAX \cite{jax2018github}.
In other words, we have essentially compiled
the probabilistic model checking problem for DTMCs
into array programs of the kind that are otherwise
ubiquitous in machine learning.
These array programs can then be optimized by
machine-learning compilers and exploit the massive parallelism
offered by hardware accelerators such as GPUs.

Finally, compiling to tensor computations makes it possible to
search for model parameters satisfying reachability properties,
via gradient-based optimization.
Since the generated array programs are differentiable,
we can compute gradients of distributional properties
with respect to model parameters using automatic differentiation,
which is readily supported by JAX.

\section{Correctness of Tensor Probabilistic Model Checking}
\label{sec:correctness}

In this section, we establish the correctness of
tensor probabilistic model checking.

In the following,
let $M$ be a \Lang model, and
let $\mu \in \Distr{\textit{State}\Bracks*{M}}$.
\cref{thm:tensor-sem-correctness} establishes the correctness of
the tensor-transformer interpretation $\tsem{M}$
with respect to the standard interpretation $\textit{Step}\sem{M}$.

\begin{theorem}
\label{thm:tensor-sem-correctness}
$\tsem{M} \, {\Xi}(\mu) = {\Xi}\parens*{\sum_{s} \mu(s) \cdot \textit{Step}\sem{M} \, s}.$
That is, the diagram below commutes.

\begin{centered}
\begin{tikzpicture}[>=Stealth, thick]

\coordinate (TL) at (0.0, 1.8);
    \coordinate (TR) at (4.0, 1.8);
    \coordinate (BR) at (4.0, 0.0);

\filldraw (TL) circle (2pt);
    \filldraw (TR) circle (2pt);
    \filldraw (BR) circle (2pt);

\node[left, xshift=-5pt] at (TL) {$M$};
    \node[right, xshift=5pt] at (TR) {}; \node[left, xshift=-5pt] at (BR) {}; \node[right, xshift=5pt] at (BR) {}; 

\draw[->, shorten >=4pt, shorten <=4pt] (TL) -- node[above] {$\sum_{s} \mu(s) \cdot \textit{Step}\sem{\cdot} \, s$} (TR);

\draw[->, shorten >=4pt, shorten <=4pt] (TR) -- node[right] {$\Xi(\cdot)$} (BR);

\draw[->, shorten >=4pt, shorten <=4pt] (TL) -- node[below left, sloped, anchor=north, xshift=-4pt,] {$\tsem{\cdot}\,{\Xi}(\mu)$} (BR);
\end{tikzpicture}
\end{centered}
 \end{theorem}

Let
$\mathcal{M}=\parens*{\textit{State}\Bracks*{M}, \textit{Init}\Bracks*{M}, \textit{Step}\sem{M}, \textit{Goal}\Bracks*{M}}$
be the DTMC represented by $M$ per \cref{sec:lang-semantics}.
Let $e$ be the goal expression specified in $M$.
Let $\mathcal{G} \defeq \textit{Goal}\Bracks*{M} = \setc{s \in \textit{State}\Bracks*{M}}{\sem{e} (s) \neq 0}$
be the set of goal states in $\mathcal{M}$.
\cref{thm:tensor-pmc-correctness} follows from
\cref{thm:tensor-sem-correctness}.

\begin{theorem}
\label{thm:tensor-pmc-correctness}
$
\Pr_M \parens*{{\Xi}(\mu) \cubeop \Diamond^{\leq n} e} = \Pr_{\mathcal{M}}\parens*{\mu\Vdash\Diamond^{\leq n} \mathcal{G}}
$
for any $n \in \mathbb{N}$.
\end{theorem}

\noindent
Proofs of \cref{thm:tensor-sem-correctness} and
\cref{thm:tensor-pmc-correctness} are available in
the technical report~\cite{tessa-cav2026-tr}.

\cref{thm:tensor-pmc-correctness-init}
establishes the ultimate correctness of
tensor probabilistic model checking with respect to
step-bounded reachability probabilities.

\begin{theorem}
\label{thm:tensor-pmc-correctness-init}
$
\Pr_M \parens*{{\Xi}(\textit{Init}\Bracks*{M}) \Vvdash \Diamond^{\leq n} e}
=
\Pr_{\mathcal{M}}\parens*{\Diamond^{\leq n} \mathcal{G}}
$
for any $n \in \mathbb{N}$.
\end{theorem}
\begin{proof} Let $\iota = \textit{Init}\Bracks*{M}$ be the initial-state distribution of the DTMC $\mathcal{M}$.
Starting from the left-hand side:
\begin{align*}
\Pr_M \parens*{{\Xi}(\iota) \Vvdash \Diamond^{\leq n} e}
&= \Pr_{\mathcal{M}}\parens*{\iota \Vdash \Diamond^{\leq n} \mathcal{G}}
\tag{by \cref{thm:tensor-pmc-correctness}} \\
&= \Pr_{\mathcal{M}}\parens*{\Diamond^{\leq n} \mathcal{G}}
\tag{by \cref{thm:reachability-via-distribution-equiv}}
\end{align*}
\qedhere
\end{proof}

\noindent
The proof of \cref{thm:tensor-pmc-correctness-init} reveals that
$\Pr_{\mathcal{M}}\parens*{\iota \Vdash \Diamond^{\leq n} \mathcal{G}}$
bridges tensor-based reachability and
the standard definition of reachability probability.

\section{Accelerating Tensor Computations with JAX and XLA}
\label{sec:accelerate}

At this point, we have reduced the verification of DTMCs (with respect
to step-bounded reachability properties) to dense tensor computations.
But that alone does not guarantee efficiency.

Efficiency hinges on how fast the tensor computations can run.
In the last decade, training and inference in machine learning (ML) have driven significant
advances in compiler optimizations for tensor computations, as well as
hardware acceleration for them.
It is thus natural to leverage off-the-shelf ML
compilers and hardware accelerators to speed up tensor probabilistic
model checking.

Specifically, \tessa compiles DTMC models to array programs in JAX \cite{jax2018github}.
JAX, embedded in Python, is a popular array programming language
for high-performance numerical computing and machine learning.
A just-in-time compiler traces the JAX program.
The resulting intermediate representation is then
handed off to XLA (Accelerated Linear Algebra) \cite{xla-github},
a domain-specific compiler designed to optimize tensor computations. 

XLA performs whole-program optimizations that are critical for
performance on modern hardware accelerators. Importantly, it
applies kernel fusion, merging multiple element-wise operations (such
as those in \cref{sec:tensor-transformer-sem,sec:tensor-pmc}) into a
single GPU kernel. Fusion significantly reduces
memory footprint and the pressure on memory bandwidth, as it
avoids materializing intermediate results between operations
to the GPU memory.
Such optimizations enable our tensor-based verification algorithm to
fully saturate the massive parallelism offered by GPUs.

Moreover, by compiling to JAX, \tessa leverages ML compiler
optimizations \emph{transparently}.
The user does not write any GPU kernel code or manage GPU memory
explicitly.
Future improvements in ML compiler technology and hardware
accelerators will directly benefit \tessa without changes to its
implementation.

\section{Evaluation}
\label{tessa:sec:eval}

We have implemented \tessa in Python.
In this section, we evaluate \tessa on two fronts: model checking and parameter synthesis.

\subsection{Model Checking}
\label{tessa:sec:model-checking-eval}

\paragraph{Evaluation methodology.}

We adopt benchmarks directly from Rubicon
\cite{holtzen2021rubicon}.
These benchmarks consist of DTMC models with dense transition dynamics.
Future work could expand the scope of the evaluation to include a more comprehensive set of benchmarks.
Nevertheless, the current selection is
already representative of challenging Markov chain verification tasks in the dense regime.

All experiments ran on a machine with an Intel Core i7-7820X CPU,
\qty{128}{GB} RAM, and an NVIDIA GeForce RTX 2080 Ti GPU (\qty{11}{GB} VRAM).
At first glance, comparing CPU-bound methods against a GPU-accelerated tool might seem like comparing apples to oranges.
However, this hardware distinction is precisely the point.
Prevailing methods rely on representations that induce
irregular memory access.
Consequently, they do not map naturally to GPUs.
In contrast, \tessa generates dense tensor workloads at which GPUs excel.
We are therefore comparing \emph{methodologies} on the hardware
they naturally map to---rather than comparing the hardware per se.

We note that the GPU used in our experiments is an older consumer-grade model
in the hardware vendor's lineup.
While we already observe substantial speedups over
state-of-the-art tools, access to frontline hardware is expected to
yield even greater~gains.

\paragraph{Baseline methods.}
For Rubicon \cite{holtzen2021rubicon}, the authors use
\storm (Sparse) and \storm (MTBDD) as
baselines in their evaluation.
We include both, as well as \dice (via the Rubicon transpiler \cite{holtzen2021rubicon}).
\begin{itemize}[leftmargin=*,topsep=2pt,itemsep=0pt,partopsep=0pt]

\item \storm (Sparse).
In this engine of the Storm model checker,
the transition dynamics of a DTMC model is represented as a
sparse matrix in the standard \emph{compressed sparse row} (CSR) format.
The engine uses efficient sparse matrix--vector multiplication
kernels from off-the-shelf libraries such as Eigen~and~Gmm++.

\item \storm (MTBDD).
This engine of Storm represents the DTMC transition dynamics
using multi-terminal binary decision diagrams.
Random access to entries of the transition matrix is not
as efficient as in the sparse-matrix engine, but the MTBDD
engine can be more memory-efficient for certain models. 

\item \dice.
In this approach, the DTMC model specification is first lowered
to the Dice probabilistic programming language~\cite{holtzen2020dice}
using the Rubicon transpiler~\cite{holtzen2021rubicon}.
The Dice compiler then represents the set of paths from
the initial state to the goal states effectively as a BDD.
This approach is shown to excel on models with certain structures.

\end{itemize}

\par
\newcommand{\Bench}[1]{{\normalfont\sffamily #1}\xspace}

\noindent
All methods are configured to use the double-precision floating-point format.

\paragraph{Results.}
We now report results for each of the benchmarks on which
the Rubicon/Dice method is evaluated \cite{holtzen2021rubicon}:
\Bench{Queues}, \Bench{Weather Factories}, and \Bench{Herman}.

\newlength{\MctEvalBodyWidth}
\newlength{\MctEvalPageHeight}
\newlength{\PlotHeight}
\newlength{\PlotWidth}
\newlength{\EvalGroupHSep}
\newlength{\EvalDummyPlotWidth}
\newlength{\EvalYlabelBoxWidth}
\newcommand{\MctSetClampedLength}[4]{\setlength{#1}{#2}\ifdim#1<#3\relax\setlength{#1}{#3}\fi
  \ifdim#1>#4\relax\setlength{#1}{#4}\fi
}
\setlength{\MctEvalBodyWidth}{\textwidth}
\setlength{\MctEvalPageHeight}{\textheight}

\begin{figure}[p]
\MctSetClampedLength{\EvalGroupHSep}{\dimexpr\MctEvalBodyWidth*4/100\relax}{6pt}{18pt}
\MctSetClampedLength{\EvalDummyPlotWidth}{\dimexpr\MctEvalBodyWidth*14/100\relax}{32pt}{84pt}
\MctSetClampedLength{\PlotWidth}{\dimexpr\MctEvalBodyWidth*30/100\relax}{70pt}{122pt}
\MctSetClampedLength{\PlotHeight}{\dimexpr(\MctEvalPageHeight - 8\baselineskip)/3\relax}{120pt}{210pt}
\MctSetClampedLength{\EvalYlabelBoxWidth}{\dimexpr\MctEvalBodyWidth/15\relax}{28pt}{44pt}
\centering
\begin{tikzpicture}
\catcode`\_=12
\def\FirstRunCol{work_seconds}
\def\SecondRunCol{measured_avg_seconds}
\def\BaselineCol{elapsed_seconds}
\def\StormAddCsv{storm.add.csv}
\def\StormSpmCsv{storm.spm.csv}
\def\RubiconCsv{rubicon.csv}
\def\TessaCsv{tessa.csv}
\def\TestQ{tessa/camera-ready/parqueues/testq/}
\def\TestH{tessa/camera-ready/parqueues/testh/}
\begin{groupplot}[
  group style={
    group size=5 by 1,
    horizontal sep=\EvalGroupHSep,
},
  height=\PlotHeight,
  width=\PlotWidth,
  legend style={
    transpose legend,
    legend columns=2,
    /tikz/every even column/.append style={column sep=-0.4ex},
},
]

\nextgroupplot[
  xlabel={$K$ ($H\,{=}\,10$)},
  legend to name=tessa:parqueues-1-A,
  ymax=160,
]

\addplot[storm-add, mark=diamond*,
]
table [
  x=Q,
  y=\BaselineCol,
  col sep=comma,
  restrict expr to domain={x}{3:9},
unbounded coords=discard,
]
{\TestQ\StormAddCsv};
\addlegendentry{\storm (MTBDD)\ \ \ \ \ }

\addplot[storm-spm, mark=triangle*]
table [
  x=Q,
  y=\BaselineCol,
  col sep=comma,
  restrict expr to domain={x}{3:10},
unbounded coords=discard,
]
{\TestQ\StormSpmCsv};
\addlegendentry{\storm (Sparse)}

\addplot[dice, mark=o,]
table [
  x=Q,
  y=\BaselineCol,
  col sep=comma,
  discard if not={status}{ok},
  discard if not A={N}{3},
  restrict expr to domain={x}{3:9},
  unbounded coords=discard,
]
{\TestQ\RubiconCsv};
\addlegendentry{\dice}

\addplot[tessa-gpu-jit-1, mark=*]
table [
  x=Q,
  y=\FirstRunCol,
  col sep=comma,
  discard if not={N}{3},
  restrict expr to domain={x}{3:11},
  unbounded coords=discard,
]
{\TestQ\TessaCsv};

\nextgroupplot[
  xlabel={$K$ ($H\,{=}\,10$)},
  ymax=8.2,
  legend to name=tessa:parqueues-1-B,
]

\addplot[tessa-gpu-jit-1, mark=*]
table [
  x=Q,
  y=\FirstRunCol,
  col sep=comma,
  discard if not={N}{3},
  discard if not={H}{10},
  restrict expr to domain={x}{3:11},
  unbounded coords=discard,
]
{\TestQ\TessaCsv};
\addlegendentry{\tessa (1st)}

\addplot[tessa-gpu-jit-2, mark=*]
table [
  x=Q,
  y=\SecondRunCol,
  col sep=comma,
  discard if not={N}{3},
  discard if not={H}{10},
  restrict expr to domain={x}{3:11},
  unbounded coords=discard,
]
{\TestQ\TessaCsv};
\addlegendentry{\tessa (2nd)}

\nextgroupplot[
  xlabel={},
  legend to name=tessa:parqueues-1-E,
ylabel={},
  axis line style={draw=none},
  ytick=\empty,
  xtick=\empty,
  width=\EvalDummyPlotWidth,
]

\addplot coordinates {(0,0)};

\nextgroupplot[
  xlabel={$H$ ($K\,{=}\,9$)},
  legend to name=tessa:parqueues-1-C,
  ymax=70,
]

\addplot[storm-add, mark=diamond*]
table [
  x=H,
  y=\BaselineCol,
  col sep=comma,
  restrict expr to domain={x}{10:10},
unbounded coords=discard,
]
{\TestH\StormAddCsv};

\addplot[storm-spm, mark=triangle*]
table [
  x=H, 
  y=\BaselineCol, 
  col sep=comma,
  restrict expr to domain={x>=100 && x<=500 || x==10}{1:1},
  unbounded coords=discard,
]
{\TestH\StormSpmCsv};

\addplot[dice, mark=o,]
table [
  x=H,
  y=\BaselineCol,
  col sep=comma,
  discard if not={status}{ok},
  discard if not A={Q}{9},
  discard if not A={N}{3},
  restrict expr to domain={x>=100 && x<=500 || x==10}{1:1},
  unbounded coords=discard,
]
{\TestH\RubiconCsv};

\addplot[tessa-gpu-jit-1, mark=*]
table [
  x=H,
  y=\FirstRunCol,
  col sep=comma,
  discard if not A={Q}{9},
  discard if not A={N}{3},
  restrict expr to domain={x>=100 && x<=500 || x==10}{1:1},
  unbounded coords=discard,
]
{\TestH\TessaCsv};

\coordinate (top) at (rel axis cs:0,1);\coordinate (bot) at (rel axis cs:1,0);

\nextgroupplot[
  xlabel={$H$ ($K\,{=}\,9$)},
  ymax=1.2,
  ymin=-0.18,
legend to name=tessa:parqueues-1-D,
]

\addplot[tessa-gpu-jit-1, mark=*]
table [
  x=H,
  y=\FirstRunCol,
  col sep=comma,
  discard if not={Q}{9},
  discard if not={N}{3},
  restrict expr to domain={x>=100 && x<=500 || x==10}{1:1},
  unbounded coords=discard,
]
{\TestH\TessaCsv};

\addplot[tessa-gpu-jit-2, mark=*]
table [
  x=H,
  y=\SecondRunCol,
  col sep=comma,
  discard if not={Q}{9},
  discard if not={N}{3},
  restrict expr to domain={x>=100 && x<=500 || x==10}{1:1},
  unbounded coords=discard,
]
{\TestH\TessaCsv};

\end{groupplot}

\node [above,inner sep=3pt,xshift=5pt] at (current bounding box.north) {\pgfplotslegendfromname{tessa:parqueues-1-A}\pgfplotslegendfromname{tessa:parqueues-1-B}};

\path (top)--(bot) coordinate[midway] (group center);
\node[above,rotate=90] at (group center -| current bounding box.west) {
  \begin{minipage}{\EvalYlabelBoxWidth}
  \centering
  \fontsize{7.5}{8}\selectfont
  time (s)
  \end{minipage}
};

\end{tikzpicture}

\caption{Scaling on the \Bench{Queues} benchmark.}
\label{fig:eval-parqueues}
 \vfill
\centering
\begin{tikzpicture}
\catcode`\_=12
\def\FirstRunCol{work_seconds}
\def\SecondRunCol{measured_avg_seconds}
\def\BaselineCol{elapsed_seconds}
\def\StormAddCsv{storm.add.csv}
\def\StormSpmCsv{storm.spm.csv}
\def\RubiconCsv{rubicon.csv}
\def\TessaCsv{tessa.csv}
\def\GeniCsv{geni.csv}
\def\TestN{tessa/camera-ready/weather-factory/testn/}
\def\TestH{tessa/camera-ready/weather-factory/testh/}
\begin{groupplot}[
  group style={
    group size=5 by 1,
    horizontal sep=\EvalGroupHSep,
},
  height=\PlotHeight,
  width=\PlotWidth,
  legend style={
    transpose legend,
    legend columns=2,
    /tikz/every even column/.append style={column sep=-0.4ex},
},
]

\nextgroupplot[
  xlabel={$N$ ($H\,{=}\,10$)},
  legend to name=tessa:weather-1-A,
ymax=460
]

\addplot[storm-add, mark=diamond*]
table [
  x=N,
  y=\BaselineCol,
  col sep=comma,
  restrict expr to domain={x}{7:11},
  unbounded coords=discard,
]
{\TestN\StormAddCsv};
\addlegendentry{\storm (MTBDD)\ \ \ \ \ }

\addplot[storm-spm, mark=triangle*]
table [
  x=N,
  y=\BaselineCol,
  col sep=comma,
  restrict expr to domain={x}{7:17},
  unbounded coords=discard,
]
{\TestN\StormSpmCsv};
\addlegendentry{\storm (Sparse)}

\addplot[dice, mark=o, ]
table [
  x=N,
  y=\BaselineCol,
  col sep=comma,
  discard if not={H}{10},
  restrict expr to domain={x}{7:17},
]
{\TestN\RubiconCsv};
\addlegendentry{\dice}

\addplot[geni, mark=+]
table [
  x=N,
  y=\BaselineCol,
  col sep=comma,
  discard if not={H}{10},
  restrict expr to domain={x}{7:16},
]
{\TestN\GeniCsv};
\addlegendentry{\geni}

\addplot[tessa-gpu-jit-1, mark=*]
table [
  x=N,
  y=\SecondRunCol,
  col sep=comma,
  restrict expr to domain={x}{7:17},
]
{\TestN\TessaCsv};

\nextgroupplot[
  xlabel={$N$ ($H\,{=}\,10$)},
  ymax=5,
  legend to name=tessa:weather-1-B,
]

\addplot[tessa-gpu-jit-1, mark=*]
table [
  x=N,
  y=\FirstRunCol,
  col sep=comma,
  restrict expr to domain={x}{7:17},
]
{\TestN\TessaCsv};
\addlegendentry{\tessa (1st)}

\addplot[tessa-gpu-jit-2,mark=*]
table [
  x=N,
  y=\SecondRunCol,
  col sep=comma,
  restrict expr to domain={x}{7:17},
]
{\TestN\TessaCsv};
\addlegendentry{\tessa (2nd)}

\nextgroupplot[
  xlabel={},
  legend to name=tessa:weather-1-E,
ylabel={},
  axis line style={draw=none},
  ytick=\empty,
  xtick=\empty,
  width=\EvalDummyPlotWidth,
]

\addplot coordinates {(0,0)};

\nextgroupplot[
  xlabel={$H$ ($N\,{=}\,13$)},
  legend to name=tessa:weather-1-C,
ymax=500,
]

\addlegendimage{storm-add, mark=diamond*}

\addplot[storm-spm, mark=triangle*]
table [
  x=H,
  y=\BaselineCol,
  col sep=comma,
  restrict expr to domain={x>=100 && x<=500 || x==10}{1:1},
  unbounded coords=discard,
]
{\TestH\StormSpmCsv};

\addplot[dice, mark=o, unbounded coords=discard]
table [
  x=H,
  y=\BaselineCol,
  col sep=comma,
  discard if not={status}{ok},
  discard if not A={N}{13},
  restrict expr to domain={x>=100 && x<=500 || x==10}{1:1},
  unbounded coords=discard,
]
{\TestH\RubiconCsv};

\addplot[geni, mark=+]
table [
  x=H,
  y=\BaselineCol,
  col sep=comma,
  discard if not={status}{ok},
  discard if not A={N}{13},
  restrict expr to domain={x>=100 && x<=100 || x==10}{1:1},
  unbounded coords=discard,
]
{\TestH\GeniCsv};

\addplot[tessa-gpu-jit-1, mark=*]
table [
  x=H,
  y=\FirstRunCol,
  col sep=comma,
  discard if not={N}{13},
  restrict expr to domain={x>=100 && x<=500 || x==10}{1:1},
  unbounded coords=discard,
]
{\TestH\TessaCsv};

\coordinate (top) at (rel axis cs:0,1);\coordinate (bot) at (rel axis cs:1,0);

\nextgroupplot[
  xlabel={$H$ ($N\,{=}\,13$)},
  ymax=2.2,
  legend to name=tessa:weather-1-D,
]

\addplot[tessa-gpu-jit-1, mark=*]
table [
  x=H,
  y=\FirstRunCol,
  col sep=comma,
  discard if not={N}{13},
  restrict expr to domain={(x>=100 && x<=500) || x==10}{1:1},
  unbounded coords=discard,
]
{\TestH\TessaCsv};

\addplot[tessa-gpu-jit-2, mark=*]
table [
  x=H,
  y=\SecondRunCol,
  col sep=comma,
  discard if not={N}{13},
  restrict expr to domain={(x>=100 && x<=500) || x==10}{1:1},
  unbounded coords=discard,
]
{\TestH\TessaCsv};

\end{groupplot}

\node [above,inner sep=3pt,xshift=5pt] at (current bounding box.north) {\pgfplotslegendfromname{tessa:weather-1-A}\pgfplotslegendfromname{tessa:weather-1-B}};

\path (top)--(bot) coordinate[midway] (group center);
\node[above,rotate=90] at (group center -| current bounding box.west) {
  \begin{minipage}{\EvalYlabelBoxWidth}
  \centering
  \fontsize{7.5}{8}\selectfont
  time (s)
  \end{minipage}
};

\end{tikzpicture}

\caption{Scaling on the \Bench{Weather Factories} benchmark.
}
\label{tessa:fig:eval-weather-factory}
 \vfill
\centering
\begin{tikzpicture}
\catcode`\_=12
\def\FirstRunCol{work_seconds}
\def\SecondRunCol{measured_avg_seconds}
\def\BaselineCol{elapsed_seconds}
\def\StormAddCsv{storm.add.csv}
\def\StormSpmCsv{storm.spm.csv}
\def\RubiconCsv{rubicon.csv}
\def\TessaCsv{tessa.csv}
\def\TestN{tessa/camera-ready/herman/testn/}
\def\TestH{tessa/camera-ready/herman/testh/}
\begin{groupplot}[
  group style={
    group size=5 by 1,
    horizontal sep=\EvalGroupHSep,
},
  height=\PlotHeight,
  width=\PlotWidth,
  legend style={
    transpose legend,
    legend columns=2,
    /tikz/every even column/.append style={column sep=-0.4ex},
},
]

\def\HermanHorizon{100}

\nextgroupplot[
  xlabel={$N$ ($H\,{=}\,\HermanHorizon$)},
  legend to name=tessa:herman-A,
  ymax=180,
  xtick={5,9,13,17},
]

\addplot[storm-add, mark=diamond*]
table [
  x=N,
  y=\BaselineCol,
  col sep=comma,
  restrict expr to domain={x}{3:17},
unbounded coords=discard,
]
{\TestN\StormAddCsv};
\addlegendentry{\storm (MTBDD)\ \ \ \ \ }

\addplot[storm-spm, mark=triangle*]
table [
  x=N,
  y=\BaselineCol,
  col sep=comma,
  restrict expr to domain={x}{3:17},
unbounded coords=discard,
]
{\TestN\StormSpmCsv};
\addlegendentry{\storm (Sparse)}

\addplot[dice, mark=o,]
table [
  col sep=comma,
  x=N,
  y=\BaselineCol,
  discard if not={status}{ok},
  discard if not A={H}{\HermanHorizon},
  restrict expr to domain={x}{3:13},
  unbounded coords=discard,
]
{\TestN\RubiconCsv};
\addlegendentry{\dice}

\addplot[tessa-gpu-jit-1, mark=*]
table [
  col sep=comma,
  x=N,
  y=\FirstRunCol,
  discard if not A={H}{\HermanHorizon},
  restrict expr to domain={x}{3:19},
unbounded coords=discard,
]
{\TestN\TessaCsv};

\nextgroupplot[
  xlabel={$N$ ($H\,{=}\,\HermanHorizon$)},
  ymax=2.4,
  legend to name=tessa:herman-B,
  xtick = {5,9,13,17},
]

\addplot[tessa-gpu-jit-1, mark=*]
table [
  x=N,
  y=\FirstRunCol,
  col sep=comma,
  discard if not A={H}{\HermanHorizon},
  restrict expr to domain={x}{3:19},
  unbounded coords=discard,
]
{\TestN\TessaCsv};
\addlegendentry{\tessa (1st)}

\addplot[tessa-gpu-jit-2, mark=*]
table [
  x=N,
  y=\SecondRunCol,
  col sep=comma,
  discard if not A={H}{\HermanHorizon},
  restrict expr to domain={x}{3:19},
  unbounded coords=discard,
]
{\TestN\TessaCsv};
\addlegendentry{\tessa (2nd)}

\nextgroupplot[
  xlabel={},
  legend to name=tessa:herman-E,
ylabel={},
  axis line style={draw=none},
  ytick=\empty,
  xtick=\empty,
  width=\EvalDummyPlotWidth,
]

\addplot coordinates {(0,0)};

\nextgroupplot[
  xlabel={$H$ ($N\,{=}\,17$)},
  legend to name=tessa:herman-C,
  ymax=450,
]

\addplot[storm-add, mark=diamond*]
table [
  x=H,
  y=\BaselineCol,
  col sep=comma,
  restrict expr to domain={x>=100 && x<=200 || x==10}{1:1},
  unbounded coords=discard,
]
{\TestH\StormAddCsv};
\addlegendentry{\storm (MTBDD)\ \ \ \ \ }

\addplot[storm-spm, mark=triangle*]
table [
  x=H,
  y=\BaselineCol,
  col sep=comma,
  restrict expr to domain={x>=100 && x<=500 || x==10}{1:1},
  unbounded coords=discard,
]
{\TestH\StormSpmCsv};
\addlegendentry{\storm (Sparse)}

\addplot[dice, mark=o,]
table [
  x=H,
  y=\BaselineCol,
  col sep=comma,
  discard if not={status}{ok},
  discard if not A={N}{17},
  restrict expr to domain={x>=100 && x<=500 || x==10}{1:1},
  unbounded coords=discard,
]
{\TestH\RubiconCsv};
\addlegendentry{\dice}

\addplot[tessa-gpu-jit-1, mark=*]
table [
  col sep=comma,
  x=H,
  y=\FirstRunCol,
  discard if not A={N}{17},
  restrict expr to domain={x>=100 && x<=500 || x==10}{1:1},
  unbounded coords=discard,
]
{\TestH\TessaCsv};

\coordinate (top) at (rel axis cs:0,1);\coordinate (bot) at (rel axis cs:1,0);

\nextgroupplot[
  xlabel={$H$ ($N\,{=}\,17$)},
  ymax=1.5,
  ymin=-0.18,
legend to name=tessa:herman-D,
]

\addplot[tessa-gpu-jit-1, mark=*]
table [
  col sep=comma,
  x=H,
  y=\FirstRunCol,
  discard if not A={N}{17},
  restrict expr to domain={x>=100 && x<=500 || x==10}{1:1},
  unbounded coords=discard,
]
{\TestH\TessaCsv};

\addplot[tessa-gpu-jit-2, mark=*]
table [
  col sep=comma,
  x=H,
  y=\SecondRunCol,
  discard if not A={N}{17},
  restrict expr to domain={x>=100 && x<=500 || x==10}{1:1},
  unbounded coords=discard,
]
{\TestH\TessaCsv};

\end{groupplot}

\node [above,inner sep=3pt,xshift=5pt] at (current bounding box.north) {\pgfplotslegendfromname{tessa:herman-A}\pgfplotslegendfromname{tessa:herman-B}};

\path (top)--(bot) coordinate[midway] (group center);
\node[above,rotate=90] at (group center -| current bounding box.west) {
  \begin{minipage}{\EvalYlabelBoxWidth}
  \centering
  \fontsize{7.5}{8}\selectfont
  time (s)
  \end{minipage}
};

\end{tikzpicture}

\caption{Scaling on the \Bench{Herman} benchmark.}
\label{fig:eval-herman}

 \end{figure}

\paragraph{\Bench{Queues}.}

The \Bench{Queues} model consists of $K$ queues, each with capacity~$3$.
Tasks arrive probabilistically at every step.
Three queues are designated type~1, while the remainder are type~2.
The goal states are those in which all type~1 queues and at least one
type~2 queue are full.
We compute the probability of reaching a goal state within
horizon~$H$.
\cref{fig:eval-parqueues} shows how each tool scales
with $K$ (left two plots) and with $H$ (right two plots).

The first plot compares all methods as $K$ varies, with $H\,{=}\,10$ fixed.
All methods scale exponentially with $K$.
On the hardest instance that any baseline method can solve ($K\,{=}\,10$),
\tessa shows
over 100$\times$ speedup over the next fastest method.

The second plot zooms in on the performance of \tessa.
Since JAX and XLA perform just-in-time (JIT) compilation, we measure
two runs to show the effect of JIT compilation:
in the second plot, the difference between the two \tessa curves
indicates the JIT compilation overhead.\footnote{Technically, the \tessa (1st) curve also includes the time taken to compile the model specification into JAX. Since this compilation happens entirely in Python, we measure it after warm-up runs and add it to the 1st-run time.}

The right two plots show how the methods scale as $H$ varies, with
$K\,{=}\,9$~fixed.
\storm (MTBDD) and \dice reach time limits at lower $K$ values.
\storm (Sparse) and \tessa scale linearly with $H$,
but \tessa is \textasciitilde40$\times$ faster at $H\,{=}\,500$.

\paragraph{\Bench{Weather Factories}.}

This model consists of $N$ factories, each in a binary state: striking
or operational. Transition probabilities are local but conditioned on
a Markov process, weather.
We verify the reachability of the state where all factories are
simultaneously striking within a given horizon~$H$.

\cref{tessa:fig:eval-weather-factory} shows how each method scales on this model.
We additionally include the \geni probabilistic programming language (PPL)
as a baseline; the benchmark is part of
its evaluation suite \cite{\CiteKeyGeni}.
The methodology of using \geni for DTMC model checking
is similar to that of \dice: both repurpose a PPL for DTMC model checking.
The difference is that \geni's compiler targets generating functions, while \dice's compiler targets BDDs.

All methods scale exponentially with $N$.
Unlike in the \Bench{Queues} benchmark, here \dice scales better than
\storm (Sparse) as $N$ increases.
On the hardest instance that any baseline method can solve ($N\,{=}\,16$),
\tessa shows over 100$\times$ speedup over the next fastest method, \geni.

All methods (that run at $N\,{=}\,13$) scale linearly with $H$.
At $H\,{=}\,500$,
\tessa shows over 100$\times$ speedup over the next fastest
method, \storm (Sparse).

\paragraph{\Bench{Herman}.}

\citeauthor{herman1990}'s protocol \cite{herman1990} is a well-known
example in the literature of
probabilistic model checking~\cite{kwiatkowska2012herman}.
It is a randomized self-stabilization algorithm for leader election in
a distributed ring of processors.
We verify the probability that a system of $N$ processors stabilizes
within a horizon of $H$ steps.
\cref{fig:eval-herman} shows how each method scales on this model.

Due to state explosion, all methods scale exponentially with $N$,
but the effect is not felt by \tessa until a larger $N$.
On the hardest instance that any baseline method can solve
($N\,{=}\,17$), \tessa shows
over 100$\times$ speedup over the next fastest method,
\storm (MTBDD).

The right two plots in \cref{fig:eval-herman} show that
\tessa scales effectively with $H$ as well.
At $H\,{=}\,500$, \tessa demonstrates over 300$\times$ speedup
over the next fastest method, \storm (Sparse).
While the speedup inherently includes the hardware advantage of a GPU,
it underscores the value of mapping the verification problem
to an accelerator-friendly representation.

\paragraph{Discussion.}
We caveat that \tessa outpaces these baseline methods for models that
\emph{fit} within the VRAM limit.
For sparse models, Storm (Sparse) and Storm (MTBDD) are
in general more space-efficient than \tessa.
Nevertheless, that \tessa achieves these speedups under the
\qty{11}{GB} VRAM constraint indicates that the method is reasonably
space-efficient for the class of models it targets
(thanks to the \tessa implementation and XLA optimizations exploiting model structure to reduce memory footprint),
whereas baseline methods may run out of memory or time out on the same dense models.

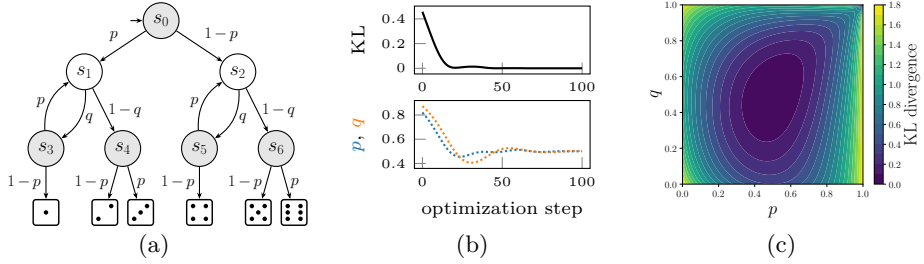
\begin{figure}[t]

\noindent
\begin{subfigure}[b]{0.33\textwidth}
\centering
\resizebox{\textwidth}{!}{\begin{tikzpicture}[
  >={Stealth[length=2.5mm]}, line width=1.2pt,
state/.style={
    circle,
    draw,
    minimum size=1.4cm,
    inner sep=0pt,
    font=\fontsize{22}{22}\selectfont,
  },
graystate/.style={
    state,
    fill=black!10
  },
edge label/.style={
    fill=none,
    inner sep=5pt,
    font=\fontsize{18}{18}\selectfont,
  },
]

\tikzset{
  die/.pic={
    \node[draw, ultra thick, rounded corners=3pt, minimum size=1cm, anchor=center] (-shape) at (0,0) {};
\foreach \x/\y in {#1} {
        \fill (\x*0.25,\y*0.25) circle (2.5pt);
    }
  }
}

\node[graystate] (s0) at (0,0) {$s_0$};
\draw[<-] (s0) -- ++(-1.2,0); 

\node[state] (s1) at (-3, -2.0) {$s_1$};
\node[state] (s2) at (3, -2.0) {$s_2$};

\node[graystate] (s3) at (-4.5, -5.0) {$s_3$};
\node[graystate] (s4) at (-1.5, -5.0) {$s_4$};

\node[graystate] (s5) at (1.5, -5.0) {$s_5$};
\node[graystate] (s6) at (4.5, -5.0) {$s_6$};

\pic (d1) at (-4.5, -7.5) {die={0/0}};

\pic (d2) at (-2.2, -7.5) {die={-1/-1, 1/1}};

\pic (d3) at (-0.8, -7.5) {die={-1/-1, 0/0, 1/1}};

\pic (d4) at (1.5, -7.5) {die={-1/-1, -1/1, 1/-1, 1/1}};

\pic (d5) at (3.8, -7.5) {die={-1/-1, -1/1, 0/0, 1/-1, 1/1}};

\pic (d6) at (5.2, -7.5) {die={-1/-1, -1/0, -1/1, 1/-1, 1/0, 1/1}};

\draw[->] (s0) -- node[edge label, above left] {$p$} (s1);
\draw[->] (s0) -- node[edge label, above right] {$1\,{-}\,p$} (s2);

\draw[->] (s1) to[bend left] node[edge label, right] {$q$} (s3);
\draw[->] (s3) to[bend left] node[edge label, left] {$p$} (s1);
\draw[->] (s1) -- node[edge label, right] {$1\,{-}\,q$} (s4);

\draw[->] (s2) to[bend left] node[edge label, right] {$q$} (s5);
\draw[->] (s5) to[bend left] node[edge label, left] {$p$} (s2);
\draw[->] (s2) -- node[edge label, right] {$1\,{-}\,q$} (s6);

\draw[->] (s3) -- node[edge label, left] {$1\,{-}\,p$} (d1-shape);

\draw[->] (s4) -- node[edge label, left] {$1\,{-}\,p$} (d2-shape);
\draw[->] (s4) -- node[edge label, right] {$p$} (d3-shape);

\draw[->] (s5) -- node[edge label, left] {$1\,{-}\,p$} (d4-shape);

\draw[->] (s6) -- node[edge label, left] {$1\,{-}\,p$} (d5-shape);
\draw[->] (s6) -- node[edge label, right] {$p$} (d6-shape);

\end{tikzpicture}}
\caption{}
\label{fig:knuth-yao-dice-dtmc}
\end{subfigure}
\hfill
\begin{subfigure}[b]{0.32\textwidth}
\centering
\definecolor{param-x}{HTML}{1f77b4} \definecolor{param-y}{HTML}{ff7f0e} \begin{tikzpicture}
\def\KydiceLossCsv{tessa/benchmarks/kydice/loss.csv}
\begin{groupplot}[
  group style={
    group size=1 by 2,
vertical sep=10pt,
  },
  width=\textwidth,
  height=0.635\textwidth,
  legend style={
/tikz/every even column/.append style={column sep=-6pt},
},
  label style={font=\scriptsize},
]

\nextgroupplot[
  ylabel={KL},
]

\addplot[black, mark=none,]
table [
  col sep=comma,
  x=step,
  y=loss,
  unbounded coords=discard,
]
{\KydiceLossCsv};

\nextgroupplot[
  xlabel={\textls[-20]{optimization step}},
  ylabel={\textcolor{param-x}{$p$}, \textcolor{param-y}{$q$}},
]

\addplot[param-x, mark=none, densely dotted,]
table [
  col sep=comma,
  x=step,
  y=params-x,
  unbounded coords=discard,
]
{\KydiceLossCsv};

\addplot[param-y, mark=none, densely dotted,]
table [
  col sep=comma,
  x=step,
  y=params-y,
  unbounded coords=discard,
]
{\KydiceLossCsv};

\end{groupplot}
\end{tikzpicture}
\caption{}
\label{fig:knuth-yao-dice-gradient}
\end{subfigure}
\hfill
\begin{subfigure}[b]{0.31\textwidth}
\centering
\raisebox{1pt}{\includegraphics[width=\textwidth]{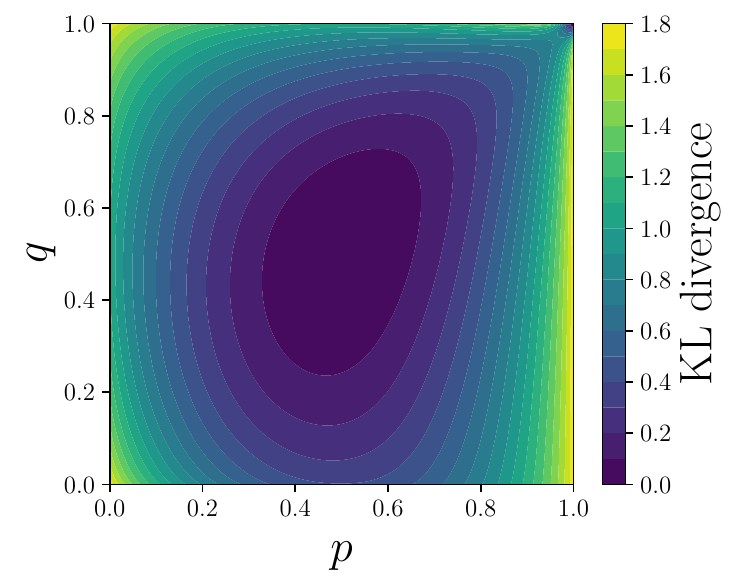}}

\caption{}
\label{fig:knuth-yao-dice-landscape}
\end{subfigure}
\caption{(a) A parametric DTMC encoding a Knuth--Yao die roller.
  In gray (resp.\ white) states, a coin of bias $p$ (resp.\ $q$) is flipped.
  (b) KL divergence and parameter values as gradient descent progresses.
  (c) Optimization landscape over parameters $p$ and $q$ as a contour plot.
  As \cref{fig:knuth-yao-dice-landscape} indicates, optimal values are
  $p\,{=}\,0.5$ and $q\,{=}\,0.5$, which are found by gradient descent
  as shown in \cref{fig:knuth-yao-dice-gradient}.
}
\label{fig:knuth-yao-dice}
\end{figure} 
\subsection{Parameter Search}

We further evaluate \tessa on parameter synthesis using the
Knuth--Yao algorithm \cite{knuthyao1976}.
This model generates a target distribution using coin flips.
\cref{fig:knuth-yao-dice-dtmc} depicts the Markov chain
\cite{jansen2022param,junges2024param}.
The task is to find the coin biases ($p$ and $q$) that produce this target distribution.
We formulate this search as an optimization problem.
The objective is to minimize the Kullback--Leibler (KL) divergence
\cite{kullback-leibler-1951}
between the model's output and the target.

\tessa compiles the model into a differentiable tensor program,
whose result can be programmatically composed with
distributional properties including but not limited to
state reachability---in this case, the KL divergence.
The resulting objective can then be directly composed with JAX's
automatic differentiation framework and
gradient-based optimizers~\cite{deepmind2020jax}.
\cref{fig:knuth-yao-dice-gradient} shows that,
starting from random initialization, gradient descent
successfully converges to the optimal parameter values well within 100 steps
in a few seconds.
This experiment highlights the \emph{flexibility} with which \tessa
can be used to optimize for distributional properties beyond
state reachability.

\section{Related Work}
\label{tessa:sec:related}
\paragraph{GPU-accelerated probabilistic model checking.}

\citeauthor{gpuprism2010bosnacki}~\cite{gpuprism2010bosnacki,bosnacki2011gpgpu} study
CUDA-accelerated sparse matrix--vector multiplication for DTMCs.
Our approach differs fundamentally from \citeauthor{bosnacki2011gpgpu}
in data representation, memory access patterns, engineering simplicity,
and the class of models targeted.
\citeauthor{bosnacki2011gpgpu}'s approach is best for DTMCs with sparse transition dynamics,
representing the transition dynamics as a flattened sparse matrix.
To mitigate the overhead of indirect memory access inherent of this storage format,
they build on the modified CSR format \cite{kwiatkowska2002outofcore}
and develop custom CUDA kernels for sparse matrix--vector
multiplication (SpMV). \citet{prismpsy2016ceska} adapt similar sparse-matrix
techniques to parameter synthesis for continuous-time Markov chains
through custom GPU kernels.
Also using sparse matrices, \citet{gpuexploreprob2025} 
perform explicit state space exploration as well as model checking entirely on the GPU.

In contrast,
our approach is less suitable for sparse models but effective for dense ones.
It compiles to JAX's dense tensor operations
and does not require GPU programming.
Rather than materializing the model into a sparse matrix (i.e., data),
we map the model to tensor transformations (i.e., code),
which XLA can fuse and optimize. This
allows our system to make good use of the high-throughput dense linear
algebra units on modern accelerators, which are often underutilized in
sparse regimes.

\citet{Bak2025} use GPUs to accelerate statistical model checking (SMC)
of extended timed automata.
SMC is fundamentally different from probabilistic model checking (PMC):
while PMC computes exact probabilities by exhaustively exploring
the state space, SMC estimates probabilities via sampling,
thus trading accuracy for feasibility.
Achieving low error margins in SMC is at the cost of high
demands on computational resources for Monte Carlo simulations.

\paragraph{State explosion.}
Techniques for mitigating state explosion
have been heavily
studied. They either compress, abstract, or prune the state space
\cite{baier1997symbolic,alfaro2000symbolic,kwiatkowska2006symmetry,katoen2007bisim,hahn2010pass,hahn2011parametric,kamaleson2016fhb,dijk2016multicore}.
State-of-the-art probabilistic model checkers (e.g., Storm~\cite{storm2022} and PRISM~\cite{prism4})
integrate such sophisticated state-space reduction techniques.

Compared to these techniques,
\tessa takes an orthogonal approach.
The translation to tensor computations
is largely oblivious to the state-explosion issue.
Rather, we rely on a tensor compiler for fusion and optimization.
In a sense, we cast state-space reduction as compiler optimizations,
offloading much of the complexity of
extracting high performance to a mature compiler~stack.

\paragraph{Distribution transformers.}
Our tensor transformer semantics is akin to the
distribution transformer semantics of \citet{kozen1981prob},
which has found many uses in the analysis of probabilistic models
(e.g., \cite{akshay2023mdp,\CiteKeyGeni}).
We recast it and establish the formal ties between the tensor transformer semantics and DTMC model checking.

\section{Conclusion}
\label{tessa:sec:conclusion}

We cast model checking of finite-horizon Markov chains
as dense tensor computations, for which
compiler and hardware support is readily available.
This new perspective delivers sizable performance gains
for models with dense transition dynamics
while maintaining mathematical soundness.
We hope our approach makes a useful addition to the toolbox of
probabilistic model checkers, extends their practical reach,
and inspires future work on tensor-based verification methods.

\makeatletter
\@ifclassloaded{llncs}{
\begin{credits}
}{}
\makeatother

\subsubsection*{Acknowledgments.}
We thank the anonymous reviewers for their valuable feedback. This work was supported in part by the Natural Sciences and Engineering Research Council of Canada. The views and opinions expressed are those of the authors and do not necessarily reflect the position of any funding agency.

\subsubsection*{Disclosure of Interests.}
The authors have no competing interests to declare that are
relevant to the content of this article.

\subsubsection*{Data-Availability Statement.}
The artifact accompanying this paper is available at \url{https://doi.org/10.5281/zenodo.19802567} \cite{tessa-cav2026-artifact}. The artifact includes the \tessa implementation, the benchmarks used in the evaluation, and instructions to reproduce the experimental results.

\makeatletter
\@ifclassloaded{llncs}{
\end{credits}
}{}
\makeatother

\ifreport
\appendix
\newpage

\allowdisplaybreaks

\section{Proofs}

\begin{theorem}
[Restatement of \cref{thm:reachability-via-distribution-equiv}]
\label{thm:reachability-via-distribution-equiv-appendix}
$
\Pr_{\mathcal{M}}\parens*{\Diamond^{\leq n} \mathcal{G}}
=
\Pr_{\mathcal{M}}\parens*{\iota \Vdash \Diamond^{\leq n} \mathcal{G}}
$.
\end{theorem}
\begin{proof}
We prove the stronger claim that for any distribution (or measure) $\mu$ over $\mathcal{S}$,
\begin{equation}
\label{eq:generalized-claim}
\sum_{s \in \mathcal{S}} \mu(s) \cdot \Pr_{\mathcal{M}}\parens*{s \vDash \Diamond^{\leq n} \mathcal{G}}
=
\Pr_{\mathcal{M}}\parens*{\mu \Vdash \Diamond^{\leq n} \mathcal{G}}.
\end{equation}
The theorem follows immediately by setting $\mu = \iota$.
We proceed by induction on $n$.

\paragraph{Base Case ($n=0$).}
From \cref{eq:reachability-via-state-base-case}, the left-hand side of \cref{eq:generalized-claim} is:
\begin{align*}
\sum_{s \in \mathcal{S}} \mu(s) \cdot \Pr_{\mathcal{M}}\parens*{s \vDash \Diamond^{\leq 0} \mathcal{G}}
&= \sum_{s \in \mathcal{G}} \mu(s) \cdot 1 + \sum_{s \notin \mathcal{G}} \mu(s) \cdot 0 \\
&= \sum_{s \in \mathcal{G}} \mu(s).
\end{align*}
This matches the definition of $\Pr_{\mathcal{M}}\parens*{\mu \Vdash \Diamond^{\leq 0} \mathcal{G}}$ in \cref{eq:reachability-via-distribution-base-case}.

\paragraph{Inductive Case.}
Assume \cref{eq:generalized-claim} holds for $n$.
We consider $n+1$.
We split the sum on the left-hand side based on membership in $\mathcal{G}$ and apply \cref{eq:reachability-via-state-inductive-case}:
\begin{align*}
\sum_{s \in \mathcal{S}} \mu(s) \cdot \Pr_{\mathcal{M}}\parens*{s \vDash \Diamond^{\leq n+1} \mathcal{G}}
&=
\sum_{s \in \mathcal{G}} \mu(s) \cdot 1
+
\sum_{s \notin \mathcal{G}} \mu(s) \cdot \parens*{ \sum_{s' \in \mathcal{S}} \eta(s, s') \cdot \Pr_{\mathcal{M}}\parens*{s' \vDash \Diamond^{\leq n} \mathcal{G}} }
\\
&=
\sum_{s \in \mathcal{G}} \mu(s)
+
\sum_{s' \in \mathcal{S}} \parens*{ \sum_{s \notin \mathcal{G}} \mu(s) \eta(s, s') } \cdot \Pr_{\mathcal{M}}\parens*{s' \vDash \Diamond^{\leq n} \mathcal{G}}.
\end{align*}
Let $\mu' = \sum_{s \notin \mathcal{G}} \mu(s) \eta(s)$.
The term in the parentheses is exactly $\mu'(s')$.
By the inductive hypothesis applied to $\mu'$, the second term becomes:
\[
\sum_{s' \in \mathcal{S}} \mu'(s') \cdot \Pr_{\mathcal{M}}\parens*{s' \vDash \Diamond^{\leq n} \mathcal{G}}
=
\Pr_{\mathcal{M}}\parens*{\mu' \Vdash \Diamond^{\leq n} \mathcal{G}}.
\]
Substituting this back, we get:
\[
\sum_{s \in \mathcal{G}} \mu(s) + \Pr_{\mathcal{M}}\parens*{\mu' \Vdash \Diamond^{\leq n} \mathcal{G}}.
\]
This matches the definition of $\Pr_{\mathcal{M}}\parens*{\mu \Vdash \Diamond^{\leq n+1} \mathcal{G}}$ in \cref{eq:reachability-via-distribution-inductive-case}.
\qedhere
\end{proof}

\newpage

\begin{lemma}
\label{lem:local-mass-enabledness}
For any module $m_k$, action $a$, and local state $s \in \STATE(Y_k)$,
\[
  \sum_{s'_k \in \STATE(X_k)} \parens*{\sem{m_k} \, a \, s}(s'_k) = \textit{En}\sem{m_k} \, a \, s.
\]
\end{lemma}

\begin{proof}
We proceed by case analysis on whether the action $a$ belongs to the alphabet of module $m_k$.

\paragraph{Case 1:} $a \notin \ACTION(m_k)$.
By definition,
\begin{align*}
  \sem{m_k} \, a \, s &= \delta_{\restrict{s}{X_k}}, \\
  \textit{En}\sem{m_k} \, a \, s &= 1.
\end{align*}
We sum the probability mass over all possible next states $s'_k$:
\[
  \sum_{s'_k \in \STATE(X_k)} \parens*{\sem{m_k} \, a \, s}(s'_k) = \sum_{s'_k \in \STATE(X_k)} \delta_{\restrict{s}{X_k}}(s'_k).
\]
The Dirac delta function $\delta$ is a point mass distribution. Its sum over the entire domain is exactly $1$:
\begin{align*}
  \sum_{s'_k \in \STATE(X_k)} \delta_{\restrict{s}{X_k}}(s'_k) = 1
  = \textit{En}\sem{m_k} \, a \, s.
\end{align*}

\paragraph{Case 2:} $a \in \ACTION(m_k)$.
By definition,
\begin{align*}
  \sum_{s'_k \in \STATE(X_k)} \parens*{\sem{m_k} \, a \, s}(s'_k)
  &=
  \sum_{s'_k \in \STATE(X_k)} \parens*{ \sum_{\Key{[} a \Key{]} \, g \,\to\, U \;\in\; \commands{m_k, a}} \indicator{\sem{g} \, s \neq 0} \cdot \parens*{\sem{U} \, s}(s'_k) }
  \\
  &=
  \sum_{\Key{[} a \Key{]} \, g \,\to\, U \;\in\; \commands{m_k, a}} \parens*{ \indicator{\sem{g} \, s \neq 0} \sum_{s'_k \in \STATE(X_k)} \parens*{\sem{U} \, s}(s'_k) }
  .
\end{align*}
The term $\sem{U} \, s$ is a probability distribution over $\STATE(X_k)$. It is defined as a convex combination of point masses where the probabilities $\theta_i$ sum to $1$. Therefore, its total mass over the state space is exactly $1$:
\[
  \sum_{s'_k \in \STATE(X_k)} \parens*{\sem{U} \, s}(s'_k) = 1.
\]
We substitute this constant back into our equation:
\[
  \sum_{s'_k \in \STATE(X_k)} \parens*{\sem{m_k} \, a \, s}(s'_k)
  = \sum_{\Key{[} a \Key{]} \, g \,\to\, U \;\in\; \commands{m_k, a}} \indicator{\sem{g} \, s \neq 0} \cdot 1
  = \textit{En}\sem{m_k} \, a \, s
  .
  \qedhere
\]
\end{proof}

\newpage

\begin{theorem}[Restatement of \cref{thm:model-distribution-preservation}]
For any $s \in \STATE$, $\textit{Step}\sem{M} \, s$ is a probability distribution.
\end{theorem}
\begin{proof}
To prove that $\textit{Step}\sem{M} \, s$ is a valid probability distribution,
we must show that the sum over all possible next states $s' \in \STATE$ equals $1$.

We consider the two cases from the definition of $\textit{Step}\sem{M} \, s$ in \cref{eq:step-interp}:

\paragraph{Case 1:} $\sum_{a} \textit{En}\sem{M} \, a \, s = 0$.
By definition, $\textit{Step}\sem{M} \, s = \delta_s$. The sum of probabilities is trivially
\[
  \sum_{s' \in \STATE} \delta_s(s') = 1.
\]

\paragraph{Case 2:} $\sum_{a} \textit{En}\sem{M} \, a \, s > 0$.
We sum the probabilities over all $s' \in \STATE$:
\begin{align*}
  \sum_{s' \in \STATE} \parens*{\textit{Step}\sem{M} \, s}(s')
  &=
  \sum_{s' \in \STATE} \frac{\sum_{a} \parens*{\bigotimes_{k=1}^{K} \sem{m_k} \, a \, {\restrict{s}{Y_k}}}(s')}{\sum_{a} \textit{En}\sem{M} \, a \, s}
  \\
  &= \frac{\sum_{a} \sum_{s' \in \STATE} \parens*{\bigotimes_{k=1}^{K} \sem{m_k} \, a \, {\restrict{s}{Y_k}}}(s')}{\sum_{a} \textit{En}\sem{M} \, a \, s}
\end{align*}

\noindent
Because the state space $\STATE$ is formed by disjoint sets of variables $X_k$, the sum of a tensor product of distributions over the joint state space equals the product of the sums of the individual marginal distributions:
\begin{align*}
  \sum_{s' \in \STATE} \parens*{\bigotimes_{k=1}^{K} \sem{m_k} \, a \, {\restrict{s}{Y_k}}}(s')
  &=
  \sum_{s' \in \STATE} \prod_{k=1}^{K} \parens*{\sem{m_k} \, a \, {\restrict{s}{Y_k}}} \parens*{\restrict{s'}{X_k}}
  \\
  &=
  \prod_{k=1}^{K} \sum_{s'_k \in \STATE(X_k)} \parens*{\sem{m_k} \, a \, {\restrict{s}{Y_k}}}(s'_k)
  \\
  &=
  \prod_{k=1}^{K} \textit{En}\sem{m_k} \, a \, {\restrict{s}{Y_k}}
  &\tag{by \cref{lem:local-mass-enabledness}}
  \\
  &=
  \textit{En}\sem{M} \, a \, s
\end{align*}

\noindent
Replacing the inner sum in the original fraction gives
\[
  \frac{\sum_{a} \textit{En}\sem{M} \, a \, s}{\sum_{a} \textit{En}\sem{M} \, a \, s} = 1
  .
  \qedhere
\]
\end{proof}

\newpage

\begin{lemma}
\label{lem:xi-linearity}
The mapping $\Xi : \UNDistr{\STATE} \to \Tensor{\STATE}$ is linear. That is, for any distributions $\mu_1, \mu_2 \in \UNDistr{\STATE}$ and scalars $c_1, c_2 \in \mathbb{R}_{\geq 0}$,
\[
\Xi(c_1 \mu_1 + c_2 \mu_2) = c_1 \Xi(\mu_1) + c_2 \Xi(\mu_2).
\]
\end{lemma}
\begin{proof}
For any state $s \in \STATE$:
\begin{align*}
&\phantom{{}={}}
\tindex{\Xi(c_1 \mu_1 + c_2 \mu_2)}{s} \\
&= (c_1 \mu_1 + c_2 \mu_2)(s)
\tag{def. of $\Xi$} \\
&= c_1 \mu_1(s) + c_2 \mu_2(s) \\
&= c_1 \tindex{\Xi(\mu_1)}{s} + c_2 \tindex{\Xi(\mu_2)}{s}
\tag{def. of $\Xi$} \\
&= \tindex{\parens*{c_1 \Xi(\mu_1) + c_2 \Xi(\mu_2)}}{s}
\end{align*}
Since this holds for all indices $s$, the tensors are equal.
\qedhere
\end{proof}

\newpage

\begin{lemma}
\label{lem:expr-equiv}
For any expression $e$,
set of variables $X$ containing the free variables of $e$,
and state $s \in \STATE(X)$,
\[
  \tindex{\tsem{e}_X}{s} = \sem{e} \, s.
\]
\end{lemma}

\begin{proof}
By structural induction on the expression $e$.
\qedhere
\end{proof}

\begin{lemma}
\label{lem:update-equiv}
For any update $u$, local state $s \in \STATE(Y_k)$, and variable $x \in X_k$,
\[
  \tindex{\tsem{u}_{Y_k} \, x}{s} = \parens*{\sem{u} \, s}(x).
\]
\end{lemma}

\begin{proof}
Let $u$ be the update $x_1 \GETS e_1 \ \cdots\ x_n \GETS e_n$. We proceed by case analysis on whether the queried variable $x$ is actively updated.

\paragraph{Case 1:} $x = x_i$ for some $i$.
\begin{align*}
  \tindex{\tsem{u}_{Y_k} \, x_i}{s} 
  &= \tindex{\tsem{e_i}_{Y_k}}{s} \tag{by \cref{eq:update-tensor}} \\
  &= \sem{e_i} \, s \tag{by \cref{lem:expr-equiv}} \\
  &= \parens*{\restrict{s}{X_k}\bracks*{x_i \mapsto \sem{e_i} s, \ldots}}(x_i) \tag{map lookup} \\
  &= \parens*{\sem{u} \, s}(x_i) \tag{by \cref{eq:update-interp}}
\end{align*}

\paragraph{Case 2:} $x \neq x_i$ for all $i$.
\begin{align*}
  \tindex{\tsem{u}_{Y_k} \, x}{s} 
  &= \tindex{\vec{I}_{x \in Y_k}}{s} \tag{by \cref{eq:update-tensor}} \\
  &= s(x) \tag{by definition of index tensors} \\
  &= \parens*{\restrict{s}{X_k}\bracks*{x_1 \mapsto \sem{e_1} s, \ldots}}(x) \tag{by $x \neq x_i$ for all $i$} \\
  &= \parens*{\sem{u} \, s}(x) \tag{by \cref{eq:update-interp}}
\end{align*}
\end{proof}

\newpage

\begin{lemma}
\label{lem:mixture-equiv}
For any mixture of updates $U$, state $s \in \STATE(Y_k)$, and next state $s' \in \STATE(X_k)$,
\[
  \tindex{\tsem{U}}{\parens*{s, s'}} = \parens*{\sem{U} \, s}(s').
\]
\end{lemma}

\begin{proof}
Let $U$ be $\theta_1 : u_1 + \dots + \theta_n : u_n$.
We expand the tensor evaluation at the joint index $\parens*{s, s'}$:
\begin{align*}
  &\phantom{{}={}}
  \tindex{\tsem{U}}{\parens*{s, s'}}
  \\
  &= \tindex{\parens*{ \sum_{i=1}^{n} \theta_i \cdot \bigodot_{x \in X_k} \indicator{ \tsem{u_i}_{Y_k} \, x \otimes \mathbf{1}_{X_k} = \mathbf{1}_{Y_k} \otimes \vec{I}_{x \in X_k} } }}{\parens*{s, s'}} \tag{by \cref{eq:mixture-tensor}} \\
  &= \sum_{i=1}^{n} \theta_i \cdot \prod_{x \in X_k} \indicator{ \tindex{\parens*{\tsem{u_i}_{Y_k} \, x \otimes \mathbf{1}_{X_k}}}{\parens*{s, s'}} = \tindex{\parens*{\mathbf{1}_{Y_k} \otimes \vec{I}_{x \in X_k}}}{\parens*{s, s'}} } \tag{point-wise evaluation} \\
  &= \sum_{i=1}^{n} \theta_i \cdot \prod_{x \in X_k} \indicator{ \tindex{\tsem{u_i}_{Y_k} \, x}{s} \cdot 1 = 1 \cdot \tindex{\vec{I}_{x \in X_k}}{s'} } \tag{point-wise equivalence} \\
  &= \sum_{i=1}^{n} \theta_i \cdot \prod_{x \in X_k} \indicator{ \parens*{\sem{u_i} \, s}(x) = s'(x) } \tag{by \cref{lem:update-equiv} and definition of index tensors} \\
  &= \sum_{i=1}^{n} \theta_i \cdot \indicator{ \sem{u_i} \, s = s' } \tag{point-wise equivalence implies state equivalence} \\
  &= \sum_{i=1}^{n} \theta_i \cdot \delta_{\sem{u_i} s}(s') \\
  &= \parens*{\sem{U} \, s}(s') \tag{by \cref{eq:mixture-interp}}
\end{align*}
\end{proof}

\newpage

\begin{theorem}
\label{thm:module-equiv}
For any module $m_k$ in model $M$,
action $a$, global state $s \in \STATE$, and global next state $s' \in \STATE$,
\[
  \tindex{\tsem{m_k} \, a}{\parens*{\restrict{s}{Y_k}, \restrict{s'}{X_k}}} = \parens*{\sem{m_k} \, a \, {\restrict{s}{Y_k}}}(\restrict{s'}{X_k}).
\]
\end{theorem}

\begin{proof}
Let $s_{in} = \restrict{s}{Y_k} \in \STATE(Y_k)$ and $s_{out} = \restrict{s'}{X_k} \in \STATE(X_k)$ be the local state projections. We proceed by case analysis on whether $a \in \ACTION(m_k)$.

\paragraph{Case 1:} $a \notin \ACTION(m_k)$.
\begin{align*}
  &\phantom{{}={}}
  \tindex{\tsem{m_k} \, a}{\parens*{s_{in}, s_{out}}} 
  \\
  &= \tindex{\parens*{ \bigodot_{x \in X_k} \indicator{ \vec{I}_{x \in Y_k} \otimes \mathbf{1}_{X_k} = \mathbf{1}_{Y_k} \otimes \vec{I}_{x \in X_k} } }}{\parens*{s_{in}, s_{out}}} \tag{by \cref{eq:module-tensor}} \\
  &= \prod_{x \in X_k} \indicator{ \tindex{\vec{I}_{x \in Y_k}}{s_{in}} = \tindex{\vec{I}_{x \in X_k}}{s_{out}} } \tag{point-wise evaluation} \\
  &= \prod_{x \in X_k} \indicator{ s_{in}(x) = s_{out}(x) } \tag{by definition of index tensors} \\
  &= \indicator{ \restrict{s_{in}}{X_k} = s_{out} } \tag{point-wise equivalence implies state equivalence} \\
  &= \delta_{\restrict{s_{in}}{X_k}}(s_{out}) \\
  &= \parens*{\sem{m_k} \, a \, s_{in}}(s_{out}) \tag{by \cref{eq:module-interp}}
\end{align*}

\paragraph{Case 2:} $a \in \ACTION(m_k)$.
\begin{align*}
  &\phantom{{}={}}
  \tindex{\tsem{m_k} \, a}{\parens*{s_{in}, s_{out}}} 
  \\
  &= \tindex{\parens*{ \sum_{\Key{[} a \Key{]} \, g \,\to\, U \;\in\; \commands{m_k, a}} \parens*{ \indicator{\tsem{g}_{Y_k} \neq 0} \otimes \mathbf{1}_{X_k} } \odot \tsem{U} }}{\parens*{s_{in}, s_{out}}} \tag{by \cref{eq:module-tensor}} \\
  &= \sum_{\Key{[} a \Key{]} \, g \,\to\, U \;\in\; \commands{m_k, a}} \indicator{\tindex{\tsem{g}_{Y_k}}{s_{in}} \neq 0} \cdot 1 \cdot \tindex{\tsem{U}}{\parens*{s_{in}, s_{out}}} \tag{point-wise evaluation} \\
  &= \sum_{\Key{[} a \Key{]} \, g \,\to\, U \;\in\; \commands{m_k, a}} \indicator{\sem{g} \, s_{in} \neq 0} \cdot \tindex{\tsem{U}}{\parens*{s_{in}, s_{out}}} \tag{by \cref{lem:expr-equiv}} \\
  &= \sum_{\Key{[} a \Key{]} \, g \,\to\, U \;\in\; \commands{m_k, a}} \indicator{\sem{g} \, s_{in} \neq 0} \cdot \parens*{\sem{U} \, s_{in}}(s_{out}) \tag{by \cref{lem:mixture-equiv}} \\
  &= \parens*{\sem{m_k} \, a \, s_{in}}(s_{out}) \tag{by \cref{eq:module-interp}}
\end{align*}
\end{proof}

\newpage

\begin{theorem}[Restatement of \cref{thm:tensor-sem-correctness}]
\label{thm:tensor-sem-correctness-appendix}
Let $M$ be a \Lang model.
Let $\mu \in \UNDistr{\textit{State}\Bracks*{M}}$.
Then,
\[\tsem{M} \, {\Xi}(\mu) = {\Xi}\parens*{\sum_{s} \mu(s) \cdot \textit{Step}\sem{M} \, s}.\]
\end{theorem}

\begin{proof}
Let $\vec{T} = \Xi(\mu)$, which means $\tindex{\vec{T}}{s} = \mu(s)$ for all $s \in \STATE$.
We evaluate $\tsem{M} \, \vec{T}$ at an arbitrary index $s' \in \STATE$:
\begin{align*}
&\phantom{={}}
\tindex{\parens*{\tsem{M} \, \vec{T}}}{s'} 
\\
&= \tindex{\parens*{\sum_{a} \vec{P}_a}}{s'} + \tindex{\parens*{\vec{T} \odot \indicator{\vec{L} = \mathbf{0}_{\Vars{M}}}}}{s'} 
\tag{by \cref{eq:step-tensor-pointfree}} \\
&= \parens*{\sum_{a} \sum_{s \in \STATE} \frac{\mu(s)}{\max\parens*{\tindex{\vec{L}}{s}, 1}} \prod_{k=1}^{K} \tindex{\tsem{m_k} \, a}{\parens*{\restrict{s}{Y_k}, \restrict{s'}{X_k}}}} + \mu(s') \indicator{\tindex{\vec{L}}{s'} = 0} 
\tag{by \cref{eq:step-tensor-contraction} and def.\ of $\odot$} \\
&= \sum_{s \in \STATE} \mu(s) \bracks*{ \sum_{a} \frac{1}{\max\parens*{\tindex{\vec{L}}{s}, 1}} \prod_{k=1}^{K} \tindex{\tsem{m_k} \, a}{\parens*{\restrict{s}{Y_k}, \restrict{s'}{X_k}}} + \indicator{\tindex{\vec{L}}{s} = 0} \indicator{s = s'} } 
\tag{by linearity of summation} \\
&= \sum_{s \in \STATE} \mu(s) \bracks*{ \sum_{a} \frac{1}{\max\parens*{\tindex{\vec{L}}{s}, 1}} \prod_{k=1}^{K} \tindex{\tsem{m_k} \, a}{\parens*{\restrict{s}{Y_k}, \restrict{s'}{X_k}}} + \indicator{\tindex{\vec{L}}{s} = 0} \delta_s(s') } 
\tag{since $\indicator{s = s'}$ equals $\delta_s(s')$}
\\
&= \sum_{s \in \STATE} \mu(s) \bracks*{ \sum_{a} \frac{1}{\max\parens*{\tindex{\vec{L}}{s}, 1}} \prod_{k=1}^{K} \parens*{\sem{m_k} \, a \, {\restrict{s}{Y_k}}}\parens*{\restrict{s'}{X_k}} + \indicator{\tindex{\vec{L}}{s} = 0} \delta_s(s') } 
\tag{by \cref{thm:module-equiv}} \\
&= \sum_{s \in \STATE} \mu(s) \bracks*{ \sum_{a} \frac{1}{\max\parens*{\tindex{\vec{L}}{s}, 1}} \parens*{\bigotimes_{k=1}^{K} \sem{m_k} \, a \, {\restrict{s}{Y_k}}}(s') + \indicator{\tindex{\vec{L}}{s} = 0} \delta_s(s') } 
\tag{since $X_1, \dots, X_K$ are disjoint sets of variables} \\
&= \sum_{s \in \STATE} \mu(s) \cdot 
\begin{cases}
\frac{\sum_{a} \parens*{\bigotimes_{k=1}^{K} \sem{m_k} \, a \, {\restrict{s}{Y_k}}}(s')}{\tindex{\vec{L}}{s}} & \text{if } \tindex{\vec{L}}{s} > 0 \\
\delta_s(s') & \text{otherwise}
\end{cases}
\tag{by analyzing cases for $\tindex{\vec{L}}{s}$} \\
&= \sum_{s \in \STATE} \mu(s) \cdot 
\begin{cases}
\frac{\sum_{a} \parens*{\bigotimes_{k=1}^{K} \sem{m_k} \, a \, {\restrict{s}{Y_k}}}(s')}{\sum_{a} \textit{En}\sem{M} \, a \, s} & \text{if } \sum_{a} \textit{En}\sem{M} \, a \, s > 0 \\
\delta_s(s') & \text{otherwise}
\end{cases}
\tag{since $\tindex{\vec{L}}{s} = \sum_{a} \textit{En}\sem{M} \, a \, s$} \\
&= \sum_{s \in \STATE} \mu(s) \cdot \parens*{\textit{Step}\sem{M} \, s}(s')
\tag{by \cref{eq:step-interp}} \\
&= \tindex{{\Xi}\parens*{\sum_{s \in \STATE} \mu(s) \cdot \textit{Step}\sem{M} \, s}}{s'}
\tag{by \cref{lem:xi-linearity}}
\end{align*}
\end{proof}

\newpage

\begin{theorem}[Restatement of \cref{thm:tensor-pmc-correctness}]
\label{thm:tensor-pmc-correctness-appendix}
For any $n \in \mathbb{N}$,
\[
\Pr_M \parens*{{\Xi}(\mu) \Vvdash \Diamond^{\leq n} e} = \Pr_{\mathcal{M}}\parens*{\mu\Vdash\Diamond^{\leq n} \mathcal{G}}.
\]
\end{theorem}
\begin{proof}
We proceed by induction on $n$.
Let $\vec{T} = \Xi(\mu)$.

\paragraph{Base Case ($n=0$).}
From \cref{eq:reachability-via-tensor-base-case}:
\[
\Pr_M \parens*{\vec{T} \Vvdash \Diamond^{\leq 0} e} = \angles{\vec{\Delta}_e, \vec{T}} = \sum_{s \in \STATE} \iverson{e}{s} \cdot \tindex{\vec{T}}{s}.
\]
Recall that $\tindex{\vec{T}}{s} = \mu(s)$ and $\mathcal{G} = \setc{s}{\sem{e}(s) \neq 0}$.
Thus $\iverson{e}{s} = 1$ if $s \in \mathcal{G}$ and $0$ otherwise.
Substituting these into the sum:
\[
\sum_{s \in \STATE} \iverson{e}{s} \cdot \mu(s) = \sum_{s \in \mathcal{G}} \mu(s).
\]
This matches the definition of $\Pr_{\mathcal{M}}\parens*{\mu\Vdash\Diamond^{\leq 0} \mathcal{G}}$ in \cref{eq:reachability-via-distribution-base-case}.

\paragraph{Inductive Case.}
Assume the theorem holds for $n$.
We expand the tensor reachability for $n+1$ using \cref{eq:reachability-via-tensor-inductive-case}:
\[
\Pr_M \parens*{\vec{T} \Vvdash \Diamond^{\leq n+1} e} =
\underbrace{\angles{\vec{\Delta}_e, \vec{T}}}_{\text{Term A}} +
\underbrace{\Pr_M \parens*{\tsem{M} \, \parens*{\vec{\Delta}_{\neg e} \odot \vec{T}} \Vvdash \Diamond^{\leq n} e}}_{\text{Term B}}.
\]
Similarly, we expand the distribution reachability using \cref{eq:reachability-via-distribution-inductive-case}:
\[
\Pr_{\mathcal{M}}\parens*{\mu\Vdash\Diamond^{\leq n+1} \mathcal{G}} =
\underbrace{\parens*{\sum_{s \in \mathcal{G}} \mu(s)}}_{\text{Term C}} +
\underbrace{\Pr_{\mathcal{M}}\parens*{\sum_{s \notin \mathcal{G}} \mu(s) \eta(s) \Vdash \Diamond^{\leq n} \mathcal{G}}}_{\text{Term D}}.
\]
From the base case logic, we know $\text{Term A} = \text{Term C}$.
It remains to show $\text{Term B} = \text{Term D}$.

Let $\mu_{\text{cont}} \in \UNDistr{\STATE}$ be the distribution of probability mass continuing from non-goal states:
\[
\mu_{\text{cont}} \defeq \sum_{s \notin \mathcal{G}} \mu(s) \delta_s.
\]

Consider the argument to $\Pr_{\mathcal{M}}$ in Term D:
\begin{align*}
\sum_{s \notin \mathcal{G}} \mu(s) \eta(s)
=
\sum_{s \notin \mathcal{G}} \mu(s) \cdot \textit{Step}\sem{M} \, s
=
\sum_{s} \mu_{\text{cont}}(s) \cdot \textit{Step}\sem{M} \, s.
\end{align*}

Now consider the tensor argument in Term B.
The masked tensor $\vec{\Delta}_{\neg e} \odot \vec{T}$ has entries:
\[
\tindex{(\vec{\Delta}_{\neg e} \odot \vec{T})}{s} = \iverson{\neg e}{s} \cdot \mu(s) =
\begin{cases}
\mu(s) & \text{if } s \notin \mathcal{G} \\
0 & \text{otherwise}
\end{cases}
\]
This is exactly the tensor representation of $\mu_{\text{cont}}$.
That is, $\vec{\Delta}_{\neg e} \odot \vec{T} = \Xi(\mu_{\text{cont}})$.
By \cref{thm:tensor-sem-correctness} (commutativity), we have:
\[
\tsem{M} (\vec{\Delta}_{\neg e} \odot \vec{T}) = \tsem{M} \, \Xi(\mu_{\text{cont}}) =
{\Xi}\parens*{\sum_{s} \mu_{\text{cont}}(s) \cdot \textit{Step}\sem{M} \, s}.
\]
Thus, the tensor argument in Term B is
${\Xi}\parens*{\sum_{s} \mu_{\text{cont}}(s) \cdot \textit{Step}\sem{M} \, s}$.
By the induction hypothesis:
\begin{align*}
&\phantom{={}}
\Pr_M \parens*{{\Xi}\parens*{\sum_{s} \mu_{\text{cont}}(s) \cdot \textit{Step}\sem{M} \, s} \Vvdash \Diamond^{\leq n} e}
\\
&=
\Pr_{\mathcal{M}}\parens*{{\sum_{s} \mu_{\text{cont}}(s) \cdot \textit{Step}\sem{M} \, s} \Vdash \Diamond^{\leq n} \mathcal{G}}.
\end{align*}
This proves $\text{Term B} = \text{Term D}$, completing the induction.
\qedhere
\end{proof}

 \clearpage
\fi

\setlength{\bibsep}{\ifreport0.5ex\else0.8ex\fi}
\makeatletter
\@ifclassloaded{llncs}{
  \bibliographystyle{tex-macros/splncs04nat.bst}
}{}
\@ifclassloaded{acmart}{
  \bibliographystyle{tex-macros/ACM-Reference-Format.bst}
}{}
\bibliography{refs.bib}

\end{document}